\documentclass[twoside, 12pt]{amsart}
\usepackage{booktabs,amsmath,amsfonts,amssymb,graphicx,bm}
\usepackage[authoryear]{natbib}
\usepackage{relsize}
\usepackage{times}
\usepackage{dsfont}
\usepackage{subfig}
\usepackage{etoolbox}
\usepackage{float}
\usepackage{epigraph}
\usepackage{ushort}
\usepackage{hyperref}
\usepackage{hyperref}
\usepackage{xcolor}
\usepackage{tikz}
\usetikzlibrary{snakes,arrows,calc, patterns, positioning, shapes.geometric, decorations.pathreplacing,decorations.markings}
\usepackage{relsize}
\usepackage{pgfplots}
\usepackage{setspace}
\usepackage[margin=14pt,font=small,labelfont=bf,
labelsep=period]{caption}

\hypersetup{
  colorlinks   = true, 
  urlcolor     = blue, 
  linkcolor    = blue, 
  citecolor   = blue 
}
\makeatletter
\def\@settitle{\begin{center}%
  \baselineskip14\p@\relax
    \bfseries
    \normalfont\Large\textbf
  \@title
  \end{center}%
}
\usepackage[left=1.25in, right=1.25in, top=1.25in]{geometry}
\usepackage[margin=10pt,font=small,labelfont=bf,
labelsep=period]{caption}
\newtheorem{theorem}{Theorem}

\newtheorem{lemma}{Lemma}
\newtheorem{proposition}{Proposition}
\newtheorem{obs}{Observation}
\theoremstyle{definition}

\newcommand{\E}{\mathbb{E}}

\def \qed {\hfill \vrule height6pt width 6pt depth 0pt}

\def\s{\sigma  }

\def \up#1{#1^{\mathsmaller{\uparrow}}}
\def\e{\epsilon  }

\def \la {{\ushort a}}
\def \ha {{\bar a}}

\usepackage{xr}
\usepackage{mathtools}

\begin{document}
\renewcommand{\thefootnote}{\fnsymbol{footnote}}

\thispagestyle{empty}
\centerline{\large \textbf{Conveying Value Via Categories}}

\vspace*{0.2in}

\centerline{Paula Onuchic \textcircled{r} Debraj Ray\footnote[2]{Onuchic: New York University, p.onuchic@nyu.edu; Ray: New York University and University of Warwick, debraj.ray@nyu.edu. We are grateful to Heski Bar-Isaac, Laura Doval, Navin Kartik, Elliot Lipnowski, Laurent Mathevet and Ennio Stacchetti for useful comments. Ray acknowledges funding under NSF grant SES-1851758. ``\textcircled{r}" indicates author names are in random order.}}

\vspace*{0.2in}

\centerline{\footnotesize{October 2020}}

\textbf{Abstract}. A sender sells an object of unknown quality to a receiver who pays his
expected value for it. Sender and receiver might hold
different priors over quality. The sender commits to
a monotonic categorization of quality. We characterize the
sender's optimal monotonic categorization. Using our characterization, we study the optimality of full pooling or full separation, the alternation of pooling and separation, and make precise a sense in which pooling is dominant relative to separation. We discuss applications, extensions and generalizations, among them
the design of a grading
scheme by a profit-maximizing school which seeks to signal student
qualities and simultaneously incentivize students to learn. Such
incentive constraints force monotonicity, and can also be embedded as a distortion of the school's prior over student
qualities, generating a categorization
problem with distinct sender and receiver priors.

\renewcommand*{\thefootnote}{\arabic{footnote}}
\section{Introduction}
A sender is about to come into possession of an object of unknown quality. Prior to knowing that quality, she commits to a \emph{categorization}. That is, she partitions the set of qualities into subsets {or \emph{categories}} --- some possibly singletons --- and verifiably commits to reveal the category in which the quality belongs. The categories  must be monotone. For instance, she can  place qualities between $a_1$ and $a_2$ into one category. She cannot, however, lump together qualities below $a_1$ or above $a_2$, where $a_2 > a_1$. Monotonicity is a natural restriction in many settings; we will discuss this assumption in some detail in Section \ref{sec:mono}.

A receiver buys the object, and pays the sender his expected value conditional on the sender's category announcement. The sender seeks to maximize expected payment.

The sender and receiver use distinct distributions to evaluate the expectation of quality. That \emph{could} mean that they hold different priors, and so disagree about the underlying distribution of qualities. But there are other possibilities. For instance, the sender might be an intermediary for individuals who differ in their optimism or pessimism about the value of the object they own, and she might be more responsive to, say, optimistic owners who are also more generous with their fees. But there are other situations with common priors that map to a reduced form with different priors.
For instance, the sender could  have state-dependent payoffs over and above her payment from the receiver: {such} payoffs can be formally mapped into a distinct prior. Or there could be incentive constraints that serve to effectively distort the measure that the sender employs to maximize expected value, \emph{even} though the priors are common. It is even possible that the resulting ``measure" isn't a distribution at all. Temporarily postponing this discussion and our last cryptic observation,  note that a difference in priors, either primitive or induced, is what makes the  problem nontrivial. With identical priors, nothing is gained or lost by categorization.

Our main result, Theorem \ref{th:1}, describes an optimal categorization.  It consists of \emph{pooling intervals} in which all qualities are declared to be in the same category and \emph{separating intervals} in which all qualities are revealed. We build an auxiliary function $H:[0,1]\rightarrow[0,1]$, where $H(x)$ is the probability, from the sender's perspective, that the quality is below the $x$ quantile in the receiver's prior. The theorem shows that an optimal categorization can be built by pooling in all intervals where $H$ differs from its lower convex envelope, and separating in all intervals where these two objects coincide. 

Theorem \ref{th:1} can be applied to study full pooling, as well as local pooling on intervals. The former is optimal when the receiver's prior dominates the sender's prior under first stochastic dominance (Proposition \ref{pr:pool}). That idea extends to local pooling (Proposition \ref{prop:poolint}). However, these intuitive findings are not paralleled for separation, explored in Propositions \ref{pr:sep} and \ref{prop:sepint}. Full separation is optimal if and only if the sender's prior dominates the receiver's in the \emph{likelihood ratio order}. In additional contrast, this dominance relationship over an interval  continues to be necessary for separation on that interval, but is not sufficient.


This asymmetry across pooling and separation has other implications. In Proposition \ref{pr:flip}, we describe a precise sense in which pooling is ``more widespread" than separation. For each set of sender and receiver priors, create a \emph{flipped} problem by switching the priors across sender and receiver. For these symmetric pairs of problems, every possible quality is pooled either in one problem or the other, but the same is not true of separation: in general, there are open sets of qualities that are pooled in both problems, and separated in neither.

In Proposition \ref{pr:alt}, we show that if both priors have densities, the optimal monotonic signal involves alternate zones of pooling and separation.\footnote{Without differentiability, it is possible to have adjacent pooling intervals without separation in between.} If we view the pools as discrete grades assigned to a continuum of potential qualities, then the separating regions that lie between two grades may be interpreted as a deliberate choice to render precise descriptions of qualities that lie close to the grade boundaries.
 
There are several applications of the categorization problem. Intermediaries often provide consumers with information about the quality of purchased goods or services. For instance, financial rating agencies classify assets according to  riskiness, certifying companies underwrite eco-friendly labels, bond issues are rated by agencies, the Department of Health provides restaurants with sanitary inspection grades, and schools grade students according to their academic achievements.  In Section \ref{sec:app}, we discuss some applications that illustrate how different sender-receiver priors could emerge from more primitive environments with common priors. Among these, we conduct an extended analysis of educational grade design, because this application simultaneously delivers monotonicity and non-common priors, two core features of our more abstract model.

Educational grades play a dual role: they signal ability, and they also incentivize learning, which could have intrinsic value over and above ability. In this sense, the problem is a mixture of information \emph{and} mechanism design, and goes beyond education. For instance, Moody's rating structure allows lenders to learn about the inherent risks of bond issuers, while incentivizing issuers to take less risk. One might wonder how these concerns are nested by our admittedly spartan model, which has no moral hazard. We show that such a setting could generates constraints that force monotonic categorization. Non-monotonic categories would not satisfy the incentive constraint  --- higher-ability students could always put in less effort and get the same grade as their lower-ability counterparts.

Furthermore, any education design model must make explicit what schools (or senders) are trying to maximize. In our case it is tuition revenue. That revenue is related to the surplus a school can extract from its students, without violating their participation constraints. Here, the receiver (an employer) pays the \emph{student}, not the school, and the student pays the school. However, what a student is willing to pay the school depends on her own prior regarding ability. If different students have different priors, the ``lowest belief type" among admitted students must pin down the equilibrium tuition rate. It is this belief that effectively becomes the prior of the school, even though its own beliefs --- and those of the receiver --- might coincide with the population distribution of abilities. 

Proposition \ref{pr:obj*} folds all these considerations into a special case of our  model, with an appropriately induced sender prior. Sender-receiver priors become different (despite being the same in the original setting), and as mentioned earlier, the induced sender prior may not be a cdf at all. This nesting generates additional insight into the grade design problem. Proposition \ref{pr:fullp} uses Proposition \ref{pr:pool} on full pooling to describe when a school  chooses to keep  quality unrevealed by a deliberate policy of coarse grading. These conditions are more propitious when learning has little intrinsic value (see also Lizzeri 1999), when the effort costs of schooling are accounted for in the optimal tuition rate (students or their parents may incompletely account for these), and when the lowest belief type has more information about her own ability. In a Supplementary Appendix, we also explore a special case in which the school optimally uses a \emph{lower censorship} categorization, where all abilities below some threshold are pooled, and all abilities are separated thereafter, each active at different points of the ability and learning distribution. 

Section \ref{sec:model} introduces the model and states our main result. Section \ref{sec:app} discusses several applications, including the one to educational grades. Section \ref{sec:mono} discusses our central monotonicity restriction. Section \ref{sec:nl} contains some remarks on a further generalization of our model to nonlinear payoff functions. Section \ref{sec:lit} has a detailed literature review. All proofs are contained in Section \ref{sec:proofs}, and several additional results indicated in the paper are collected in a Supplementary Appendix.

\section{Model\label{sec:model}}
A sender (female) is about to come into possession of an object of  quality $a$, ``distributed" according to some function $S$.  A receiver (male) buys the object  and obtains value equal to its quality. He believes that quality is distributed according to a continuous cdf $R$, strictly increasing on $[\la,\bar{a}]$ with $R(\la) = 0$ and $R(\bar a) =1$. He stands ready to pay his expected value for the object, where expectations are computed using $R$ and any information that the sender might have chosen to reveal.

One interpretation is that sender and receiver hold distinct priors. But there are others, including the possibility that $S$ and $R$ are reduced forms of an extended model, possibly one with common priors. For instance, the schooling application in Section \ref{sec:educ} has a common prior, but generates a special case in which $S$ is not only distinct from $R$, but isn't a cdf at all. So while it is convenient for the exposition to view $S$ as a  cdf, our results do not rely on such a presumption. We assume, instead, that $S$ has bounded variation on the same support $[\la,\bar{a}]$ as $R$, with $S(\la)$ finite and $S(\bar{a})=1$, and that it is left-continuous with only upward jumps. It would indeed be a cdf if it were nondecreasing with $S(\la)=0$,\footnote{The dominant convention for cdfs ask for right continuity, but our (equally valid) convention eases notation. The assumed absence of downward jumps also eases notation and some  technicalities, but is  dispensable.} but to accommodate some applications, we do not formally make these assumptions.

Prior to learning quality, the sender commits to a monotonic \emph{categorization}, revealing the (possibly degenerate) interval in which quality lies. For instance, she can  create two categories for qualities between $\la$ and $a_1$ and qualities between $a_1$ and $\bar{a}$. She cannot, however, lump together qualities below $a_1$ or above $a_2 > a_1$ into one category without also including all qualities between $a_1$ and $a_2$ in that same category. The sender chooses  categories so as to maximize her expected revenue from the sale of the object. 

Formally, let $\mathcal{P}\subset(\la,\bar{a})$ be an open set. Then $\mathcal{P}$ can be written as  a countable collection of disjoint intervals of the form $(p,p')\in[\la,\bar{a}]$. Define:
\begin{equation}
A(x)=
\begin{cases}
\mathbb{E}_R\left[ y|y\in[p,p')\right]\text{, if }x\in[p,p')\text{ for some }(p,p')\in\mathcal{P}\\
x\text{, otherwise.}
\end{cases}
\label{eq:A}
\end{equation}

Let  $\mathcal{A}_R$ be the collection of all such functions.
This set is nonempty; e,g., $A(x) =  \mathbb{E}_R(a)$ for all $x$ lies in it.
The sender picks $A \in \mathcal{A}_R$ to 
\begin{equation}
\mbox{maximize } \int_\la^{\bar{a}} A(x)dS(x),
\label{eq:maxa}
\end{equation}
noting that this Stieltjes integral is just expected value under $S$ when $S$ is a cdf.\footnote{The Stieltjes integral is well-defined because $S$ has bounded variation and is left continuous, while the integrand $A$ is right continuous.}
Observe that any $A \in \mathcal{A}_R$ has (at most) countably many disjoint intervals on which it is constant; these are pooling intervals. Adjacent pooling intervals have distinct constant images. Qualities that are not in pooling intervals are in separating regions. By convention, we always separate $\ha $ (so $A(\ha ) = \ha $). All pooling intervals are closed on the left and open on the right.  These choices are without loss of generality because (a) $R$ has no mass points, and (b) $S$ is increasing at any mass point, and so the sender would prefer to close any discontinuity in $A$ on the right rather than the left.

%
%

The categories implicit in any $A \in \mathcal{A}_R$ consists of all the intervals, and \emph{every element} of every ``separating" interval or singleton. Once told the categorization and the particular category in which the object lies, the conditional value of the object --- in the receiver's eyes --- is given by the expectation of $a$ in that category element, under $R$. This is precisely equation (\ref{eq:A}).  
To illustrate, let $R$ be uniform on $[0,1]$, and $S$ a uniform cdf on an interval of length $2\e$ with mean $\{ 3/4\}$. If the sender fully pools, then the associated $A(x)$ has constant value $1/2$ on $[0,1)$, since the expectation is calculated under $R$ with no additional information. The sender then integrates over $A$ with respect to $S$, and so her value under this policy is $1/2$, which is lower than the expectation of $x$ under $S$. She could achieve the latter by committing to reveal quality; i.e., by choosing $A(x) = x$. 

The sender can do still better, though. If she pools all qualities in the two intervals $\left[ 0,\frac{3}{4}-\e\right)$ and  $\left[ \frac{3}{4}-\e,1\right)$, then $A(x)=\frac{3}{8}-\frac{\e}{2}$ if $x$ lies in the former interval and $A(x)=\frac{7}{8}-\frac{\e}{2}$ if $x \in \left( \frac{3}{4}-\e, 1\right)$. The sender then integrates over $A$ with respect to $S$, obtaining $\frac{7}{8}-\frac{\e}{2}$.

{\subsection{Optimal Categorization\label{sec:uniform}} 
For $A \in \mathcal{A}_R$, each pooling interval is written as $[p, p')$, where $p$ is the left edge and $p'$ the right edge. Define the \emph{weighting} $\Psi $ associated with $A\in\mathcal{A}_R$ by:
\begin{figure}[t!]
\begin{tikzpicture}[>=stealth, scale=.8]
    \begin{axis}[thick,
        xmin=0,xmax=1.08,
        ymin=0,ymax=1.08,
        axis x line=middle,
        axis y line=middle,
        axis line style=->,
        x label style={at={(axis description cs:0.5,-0.1)},anchor=north},
        ticks=none
        ]
 \addplot[no marks,very thick, -] expression[domain=0:1,samples=200]{1-(1-x^(1/2))/(1-.9*x)} 
                    node[pos=1,anchor=west]{$S$}; 
    \end{axis}
	\node at (0,0) [below]{$t_0=0$};
	\node at (0,0)[circle,fill,inner sep=1pt]{};
	\draw[black,snake, very thick] (0,0)--(6.37,0);
	\node at (6.37,0) [below]{$\bar{a}$};
	\node at (6.37,0)[circle,fill,inner sep=1pt]{};

    \begin{scope}[xshift=9cm]
       \begin{axis}[thick, 
        xmin=0,xmax=1.08,
        ymin=0,ymax=1.08,
        axis x line=middle,
        axis y line=middle,
        axis line style=->,
        x label style={at={(axis description cs:0.5,-0.1)},anchor=north},
        ticks=none
        ]
 \addplot[no marks,very thick, -] expression[domain=0:1,samples=200]{1-(1-x^(1/2))/(1-.9*x)} 
                    ; 
  \addplot[no marks, very thick, red, -] expression[domain=0:0.4,samples=100]{((1-(1-.4^(1/2))/(1-.9*.4))/.4)*x};
  \addplot[no marks, very thick, red, -] expression[domain=.4:1,samples=100]{1-(1-x^(1/2))/(1-.9*x)} 
  node[pos=1,anchor=west]{{\color{black}$\Psi$}};
    \end{axis}
    \node at (0,0) [below]{$t_0=0$};
	\node at (0,0)[circle,fill,inner sep=1pt]{};
	\node at (2.48,0) [below]{$t_1$};
	\node at (2.48,0)[circle,fill,inner sep=1pt]{};
	\draw[black,snake, very thick] (2.48,0)--(6.37,0);
	\draw[dashed, thin](2.48,0)--(2.48,2.37);
	\node at (6.37,0) [below]{$\bar{a}$};
	\node at (6.37,0)[circle,fill,inner sep=1pt]{};
        \end{scope}
\end{tikzpicture}

\caption{\small{$R$ (not shown) is uniform on $[0, \bar{a}]$ and $S$ is as pictured. Jagged lines indicate separating regions. In the first panel, there is full separation and $\Psi = S$. In  the second panel, the sender pools on $[t_0, t_1)$ and separates elsewhere, and $\Psi$ is shown in red.}}
\label{fig:succ}
\end{figure}
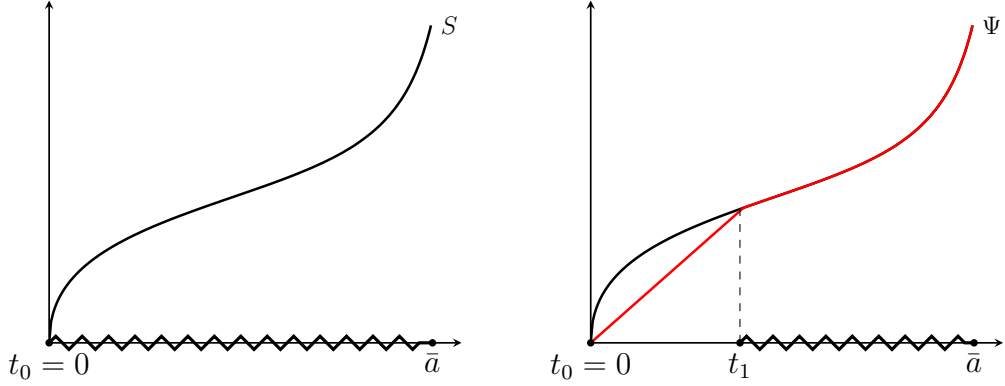
\begin{equation}
\scalebox{0.9}{$\displaystyle{\Psi (x, A )}$}=\begin{cases}
\scalebox{0.9}{$\displaystyle{S(t)+\left[ R(x)-R(t)\right]\left[\frac{S(t')-S(t)}{R(t')-R(t)}\right]\text{  if }x \mbox{ is in a pooling interval $[p,p')$}}$};\\[0.5em]
\scalebox{0.9}{$\displaystyle{S(x)}$ \text{  otherwise.}}
\end{cases}
\label{eq:psi}
\end{equation}
With $R$ strictly increasing, $\Psi$ is well-defined. Also, $\Psi (\la, A) = S(\la)$ and $\Psi (\bar a, A)=1$. If $S$ is a cdf, so is $\Psi$.  If $S$ has bounded variation, so does $\Psi$.

Next, for any continuous function $H:[0,1]\rightarrow \mathbb{R}_+$, define its \emph{lower convex envelope} by
\[
\breve{H}(x) \equiv \min\{y | (x,y)\in \mbox{Co}(\mbox{Graph}(H))\},
\]
This chalks out --- uniquely --- the largest convex function we can place below $H$.
In what follows, we study the particular function $H = S\circ R^{-1}$. 
\begin{obs}
\textnormal{(i)} The value to the sender, \scalebox{1}{$\displaystyle{\int_\la^{\bar{a}}A(x)dS(x)}$}, equals \scalebox{1}{$\displaystyle{\int_\la^{\bar{a}} xd\Psi (x, A)}$}.

\textnormal{(ii)}  Furthermore, for any $A\in\mathcal{A}_R$ and $z \in [0,1]$,  $\Psi (R^{-1}(z), A)\geqslant \breve{H}(z)$.
\label{obs:v}
\end{obs}
Part (i) states that the sender's value under any categorization $A$ is  found by integrating $x$ over $[\la,\bar{a}]$ under the weighting $\Psi(\cdot, A)$, which in separating regions ``follows" the sender's prior $S$ and in pooling regions ``follows" the receiver's prior $R$. See Figure \ref{fig:succ}. Moreover, integration by parts reveals that 
\begin{align}
\label{eq:ibp} 
\int_\la^{\bar{a}} xd\Psi (x, A)&= [1- \Psi(\bar a, A)]{\bar a}-[1- \Psi(\la,A)]\la+\int_\la^{\bar{a}}(1-\Psi(x,A))dx \nonumber \\   
&= -[1-S(\la)]\la+\int_\la^{\bar{a}}(1-\Psi(x,A))dx,
\end{align}
where we use $\Psi(\la,A) = S(\la )$ and $\Psi(\bar a, A) = 1$. Therefore we will have found our optimal categorization if we can find a suitable pointwise lower bound to every $\Psi $. That motivates part (ii), which connects the weighting $\Psi $ to the lower convex envelope of $H=S\circ R^{-1}$. Notice that $\breve{H}$ has zones where it coincides with locally convex segments of $H$ (not necessarily all of them), and other intervals where it is a straight line connecting two points of the form $(z,H (z))$ and $(z',H (z'))$. We can thus fashion a categorization by pooling the  intervals where $\breve{H}$ is a straight line and by separating everywhere else.
\begin{obs}
There exists $A^*\in\mathcal{A}_R$ such that $\Psi (R^{-1}(z), A^*)=\breve{H}(z)$ for all $z\in[0,1]$.
\label{obs:ex}
\end{obs}
Together, Observations \ref{obs:v}(ii) and \ref{obs:ex}  imply that there is a categorization $A^*$ such that for every categorization  $A\in\mathcal{A}_R$ and $z \in [0,1]$,  $\Psi (R^{-1}(z), A) \geqslant \Psi (R^{-1}(z), A^*)$. In other words,
\begin{equation}
\Psi(x,A)\geqslant \Psi(x,A^*) \mbox{ for every } x \in [\la,\bar a] \mbox{ and } A\in\mathcal{A}_R.
\label{eq:fosd}
\end{equation}
Consequently, (\ref{eq:ibp}) and (\ref{eq:fosd}) implies that  sender value is weakly higher under $A^*$ than under any $A\in\mathcal{A}_R$, which yields
\begin{theorem}
A solution to the sender's problem exists. A categorization $A^*$ is a solution if and only if  $\Psi(R^{-1}(z),A^*)=\breve{H}(z)$ for all $z\in[0,1]$, where $H = S \circ R^{-1}$.
\label{th:1}
\end{theorem} 
Notice that finding the lower convex envelope of $S \circ R^{-1}$ is equivalent to finding the upper concave envelope of the inverse $R \circ S^{-1}$, and so this characterization would be equivalent to the \emph{concavification} of this inverse, if $S$ were invertible. However, because $S$ isn't even assumed to be a cdf, we cannot (and do not need to) make this assumption.
}

\begin{figure}[t!]
\begin{center}
\hspace*{-0.25in} 
\subfloat[{\small Concave Receiver Prior}]{
\includegraphics[width=0.43\textwidth]{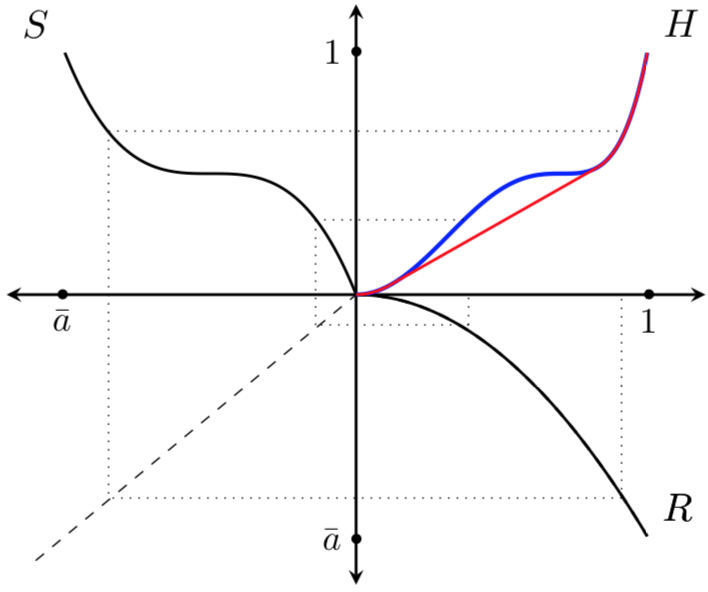}} 
\subfloat[{\small Convex Receiver Prior}]{
\includegraphics[width=0.43\textwidth]{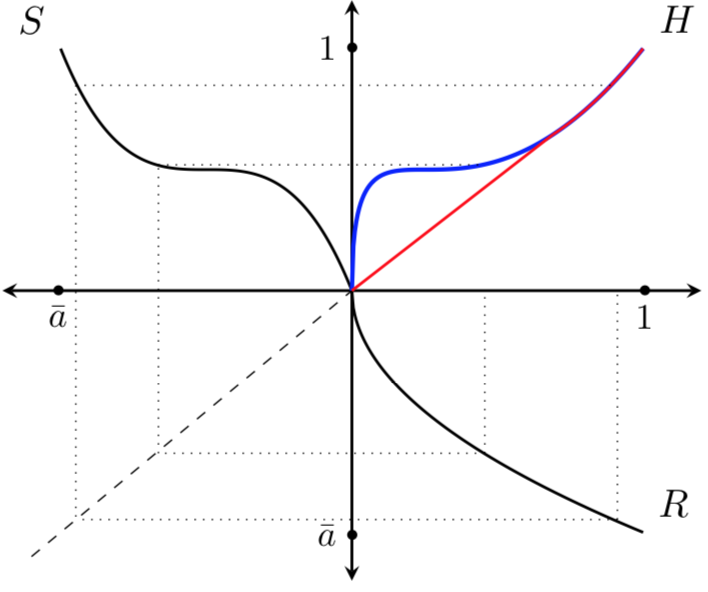}}
\end{center}
\begin{caption}
{\small{The construction of  $H = S\circ R^{-1}$ and its lower convex envelope. In panel A,  $R$ is concave. In panel B, $R$ is convex. In both cases, $S$ has a reverse-logistic shape. $H$ is shown in blue and its lower convex envelope in red.}\label{fig:H}}
\end{caption}
\end{figure}

Figure \ref{fig:H} displays the optimal categorization. $S$ (shown horizontally flipped) has a ``reverse-logistic" shape. In panel A, $R$ is concave (shown vertically flipped) and in panel B, $R$ is convex (again vertically flipped). In each panel, the function $H = S\circ R^{-1}$ is derived, and displayed in the north-east quadrant. Its lower convex envelope is also displayed.
To see informally how this categorization behaves, note that if $R$ is concave (panel A), this compresses the concave segment of the derived $H$, and elongates the convex segment. The resulting pool is therefore ``small." This is reasonable: concavity of $R$ signals receiver pessimism about quality, so it is better that the sender separate qualities to a greater degree (relative to uniform $R$). In fact, Panel A shows two distinct zones of separation. In panel B, $R$ is convex --- the receiver is relatively optimistic. That accentuates the concave segment of $H$ and induces greater pooling, which is intuitive in the face of that optimism. 

Under mild additional assumptions, it is possible to write down an algorithmic procedure that generates the pooling and separating intervals implicit in the characterization of Theorem \ref{th:1}. The procedure is connected to the solution method in Rayo (2013), though the connection to our characterization is nontrivial; see Supplementary Appendix for details.

\subsection{Full Pooling and Pooling Intervals}  Theorem \ref{th:1} has implications for full pooling and full separation, and can also be used to study pooling and separation on sub-intervals. Headway on the latter is  limited by the fact that  optimal categorization has  global features which cannot be solved by local inspection. We begin with pooling. The proposition states that if the receiver is more optimistic than the sender in the sense of first order stochastic dominance or ``fosd"), then the sender should not reveal any information that might make the receiver learn the more pessimistic truth --- and so full pooling on $[\la , \ha)$ is optimal.\footnote{As already noted, we always separate $\ha $ by convention; this makes no difference to the analysis.}
\begin{proposition}
Full pooling on $[\la , \ha)$ is optimal if and only if $S(x)-S(\la)\geq [1-S(\la)]R(x)$ for all $x\in[\la,\bar{a}]$. If $S$ is a cdf with $S(\la)=0$, then this condition is equivalent to $R$ fosd $S$. 
\label{pr:pool}
\end{proposition}
 This intuitive idea extends to local regions over which the receiver is more optimistic than the sender; those are natural intervals over which the sender should pool. 
\begin{proposition} Let $S$ be a cdf. 
If $R$ fosd $S$ on $[a,b]$, then there exists an optimal categorization where $[a,b)$ belongs to a pooling interval. 

Conversely, if $[a,b)$ belongs to a pooling interval, then $R$ fosd $S$ on that pooling interval.
\label{prop:poolint}
\end{proposition}

\subsection{Full Separation and Separating Intervals} The same intuitive ideas do not carry over in symmetric fashion to separation, though Theorem \ref{th:1} can still be used. Unlike the case of full pooling, first-order stochastic dominance is not the relevant criterion. Full separation is optimal if and only if $S$ dominates $R$ in the \emph{likelihood ratio order}.\footnote{$S$ dominates $R$ in the likelihood ratio order if $\mathbb{P}_S\left[x\in X\right]\mathbb{P}_R\left[x\in Y\right]\leqslant \mathbb{P}_S\left[x\in Y\right]\mathbb{P}_R\left[x\in X\right]$ for all measurable sets $X$ and $Y$ with $X\leqslant Y$, where $X\leqslant Y$ means that $x\in X$ and $y\in Y$ imply $x\leqslant y$. Shaked and Shanthikumar (2007) show that this is equivalent to $S\circ R^{-1}$ being convex (see Equation (1.C.4). In Theorem 1.C.5, Shaked and Shanthikumar (2007) also show that $S$ dominating $R$ in the likelihood ratio order is equivalent to $S$ fosd $R$ over every subinterval of $[\la,\bar{a}]$.} Put another way, the converse of full pooling is not separation. It is true that if $S$ fosd $R$, she would gain from splitting $[\la,\bar{a}]$ into two or more pools. But for \emph{full} separation to be optimal, she must want to split \emph{every} pool. More formally:
\begin{proposition}
Full separation is optimal if and only if $H=S\circ R^{-1}$ is convex on $[0,1]$. If $S$ is a cdf, then this condition is equivalent to S dominating $R$ in the likelihood ratio order. 
\label{pr:sep}
\end{proposition}
Unlike Proposition \ref{prop:poolint}, this assertion admits only a partial extension to sub-intervals: $S$ dominating $R$ in the likelihood order over $[a,b)$ is a necessary but not sufficient condition for $[a,b)$ to be nested in a separating interval. 
\begin{proposition} Let $S$ be a cdf.
If $[a, b)$ is a subset of some separating interval of qualities in any optimal categorization, then $S$ dominates $R$ in the likelihood ratio order on $[a,b)$.
\label{prop:sepint}
\end{proposition}

\subsection{Separation Amidst Pooling} This subsection and the next turn to mixed solutions with both pooling and separation. We observe, first, that pooling and separation alternate under additional smoothness conditions on $S$ and $R$.\footnote{Without smoothness, it is possible to write an example with adjacent pooling intervals and no separating zone in between; see Panel A of  Figure \ref{app-fig:scallops} in the Supplementary Appendix for an illustration.}
\begin{proposition}
Suppose $R$ and $S$ are differentiable on $[\la, \ha]$, with strictly positive densities on $(\la,\bar{a})$. Then between two pooling intervals there is always a zone of separation.
\label{pr:alt}
\end{proposition}
Viewing pools as discrete grades assigned to a continuum of potential qualities,  it is as if a cross-pool boundary gets  precise annotations, so that optimal categorization involves alternation over coarseness and revelation. A plus or minus annotation in student grades loosely conforms to this description, or letters of recommendation in which students at either extreme of a  grade are flagged. In finance,  rating agencies often report ``outlooks'' on assets, signaling that a rating is about to go up --- this asset is among the best in its category --- or down --- this asset is among the worst in its category. In some ratings, additional filters often provide lists which point out ``the best purchase under \$x," and so on. As noted, this observation does rely on smoothness in the sender and receiver  distributions.\footnote{Rayo (2013) also finds that alternating pooling and separating intervals is optimal. For more on the connections with Rayo's work, see Section \ref{sec:lit} and the Supplementary Appendix.}

\subsection{The Predominance of Pooling\label{sec:pre}}
Despite this alternating pattern,  there is a precise sense in which pooling is ``more widespread" than separation. Consider a problem with sender prior $S$ and receiver prior $R$. Now ``flip" the  priors so that $S' \equiv R$ and receiver prior $R' \equiv S$.  Note that if full separation is  optimal in the original problem, then full pooling is optimal in the flipped problem. For by Proposition \ref{pr:sep}, we know that $S$ dominates $R$ in the likelihood ratio order, which implies that $S$ fosd $R$. By Proposition \ref{pr:pool}, the condition for full pooling is met in the flipped problem. However, the converse is not true --- optimal full pooling in the original problem does \emph{not} necessarily imply full separation in the flipped problem. If $S$ and $R$ are such that full pooling is optimal in the original problem, then $R$ fosd $S$, but that does not necessarily imply that $R$ dominates $S$ in the likelihood ratio order, which is the requirement for full separation to be optimal in the flipped problem. 

This argument can be significantly extended. Say that a function $f$ on a real interval is \emph{nowhere locally affine} if there is no sub-interval $[z_1, z_2]$ with $f(z)$  affine on $[z_1, z_2]$.
\begin{proposition}
If $R$ and $S$ are increasing and continuous, every quality lies in some  pooling interval of an optimal categorization in either the original problem or the flipped problem. 

Furthermore, if $R$ and $S$ are differentiable on $[\la, \ha]$ with strictly positive densities on $(\la,\bar{a})$,  if $H=S\circ R^{-1}$ is neither globally convex nor concave, nor locally affine, then there is an open interval of qualities that lie in optimal pooling zones for both problems, but cannot be part of optimal separating zones for either problem. 
\label{pr:flip}
\end{proposition}
So every quality is pooled  in at least one of the two problems. 
When $H=S\circ R^{-1}$ is globally convex or concave, Propositions \ref{pr:pool} and \ref{pr:sep} tell us that we have full separation in one problems, and full pooling in the other --- in short, a tie between pooling and separation. But in all other cases --- modulo the technical restrictions of differentiability and absence of local flats\footnote{These  restrictions are for ease of exposition. Dropping them necessitates a  more qualified statement.} --- the first statement of the proposition continues to hold, and in addition there are qualities that are pooled in \emph{both} problems, and therefore separated in neither.

\section{Some Applications\label{sec:app}} 
Category design is a problem that arises in many settings; see the Introduction for some examples. One could, of course, take the non-common prior literally --- sender and receiver simply  agree to disagree.\footnote{For instance, Yildiz (2004), Van den Steen (2004, 2009, 2010, 2011) and Che and Kartik (2009)  consider  environments in which agents hold different priors that affect their incentives to acquire information.} Or perhaps the sender faces receivers with different priors (that aggregate to the sender's prior) because they have different experiences or information. She might be interested in a robust
 solution --- one that generates the highest return to her under the most pessimistic receiver prior in some exogenously given class.  We don't dwell on either of these obvious interpretations,  but emphasize other settings with \emph{common} beliefs that can be  mapped into reduced-form models with heterogeneous priors. 

For instance, while priors are common, the sender might derive a payoff from the state over and above receiver payments. Suppose that sender value under quality $a$ and receiver's posterior mean $A$ is $U(a,A) = \alpha (a)A + \beta (a)$, which is a \emph{linear} function of $A$, but potentially a nonlinear function of the state $a$. In that case, the sender picks $A\in\mathcal{A}_R$ so as to
\begin{equation*}
\text{ maximize } \int_\la^{\bar{a}}\alpha(a)A(a)dS(a).
\end{equation*}
But this state-dependent preference can be mapped into a new prior $\hat S$  to the sender, by letting $d\hat{S}(a)=\alpha (a)dS(a)$. In particular, even if sender and receiver start out with the same prior, their ``effective beliefs" will diverge, and this will affect the choice of pooling and separation. If, for instance, $\alpha$ is increasing,  the sender obtains more value by having $A(a)$ be larger precisely when $a$ is higher. It stands to reason that in that case, the sender chooses to reveal more information to the receiver, which leads to more separation. While we omit a formal proof, it is easy to see how this maps to our model: $\hat{S}$ --- the new sender ``prior'' just defined --- dominates $S$ in the likelihood ratio order. Consequently, the separating zones under $S$ will be a subset of the separating zones under $\hat{S}$. 

We can take a step further and think of even more primitive examples which map into the state-based-preference or divergent-prior  reduced form. Contributions such as Rayo (2013), Kartik, Lee and Suen (2016) and Kartik, Lee and Suen (2020) have also studied problems that map in this way on to ``distorted priors." 

\subsection{Some Examples that Generate State-Dependent Preferences\label{sec:exs}}
In what follows, we briefly discuss some applications that generate different reduced-form priors starting from a common-prior setting. For one more example, see the end of Section \ref{sec:nl}.

\emph{Retailing Intermediation.}  Rayo and Segal (2010) consider a retailing intermediary who is paid in some ratio to the value of a sale. The intermediary receives unit fees for different \emph{prospects}. Let $v$ be the buyer's value of a prospect and $\pi$ the unit fee paid out under that prospect. Say no two prospects have the same value, so each $v$ is associated with a single unit fee $\pi (v)$.  Let $R$ be the commonly held prior on prospects.\footnote{Rayo and Segal (2010) assume that there are only finitely many prospects, so that $R$ has finite support. In our exposition, we take $R$ to be a continuous distribution. Both our results can be extended to allow for mass points and their results can be extended to continuous distributions.} Then the intermediary's profit when she picks a categorization $A\in\mathcal{A}_R$ is
\begin{equation*}
\int_\la^{\bar{a}}\pi(v)A(v)dR(v)=\int_\la^{\bar{a}}A(v)dS(v)
\end{equation*}
where we set $dS(v)=\pi(v)dR(v)$. Note that $d [S\circ R^{-1}](v)=dS(v)/dR(v)=\pi(v)$, so that regions where $\pi$ is increasing are equivalent to regions where $S\circ R^{-1}$ is convex, and regions where $\pi$ is decreasing are equivalent to regions where $S\circ R^{-1}$ is concave.

We can therefore use our method to solve for the optimal monotonic\footnote{That is, the categorization classifies revealed  receiver values into intervals or singletons.} categorization in Rayo and Segal's environment. They characterize the optimal signal without any restriction to monotonic categorizations. In Section \ref{sec:mono}, we compare the optimal monotonic categorization found in our paper to the optimal signal characterized in their paper. 

\emph{Affirmative Action.} A public authority --- the sender --- designs the disclosure of the results of an admission test to a university. Suppose that the university admits a student with probability proportional to the perceived quality of that student. The public authority seeks to maximize weighted welfare, placing different weights on individuals based on the social group they belong to --- the idea being that the sender places higher weights on social groups that have been historically underprivileged. Let $w_G$ the weight placed on group $G$, and $F_G$ the distribution of test scores for students in that group. The relevant ``prior'' to the sender is then $S=\sum_{G}w_GF_G$. If the receiver --- the university --- only cares about the underlying test score, the appropriate receiver ``prior'' is $R=\sum_{G}F_G$.

\emph{Peer Effects.} 
A school --- the sender --- wants to design different educational ``tracks," such as honors programs, when there are peer effects across students. Specifically, the school seeks track-specific thresholds, with students admitted to all tracks for which their abilities exceed the track thresholds. Let $R$ on $[\la,\bar{a}]$ be the ability distribution among students at the start of the year. Suppose that a student with initial ability $a$ is placed in a track where the average initial ability of the students is $A(a)$. Then suppose that the student's end-of-year ability is given by $\hat{a}(a)=\lambda_1(a)+\lambda_2(a)A(a)$. The function $\lambda_1$ describes how ability ``naturally'' evolves through the year, while $\lambda_2$ measures the influence of peer effects on a student of initial ability $a$. Say $\lambda_1$ is  increasing and $\lambda_2$ is always weakly positive. The school's objective is to maximize end-of-year student abilities. They will then pick a categorization so as to maximize
\begin{equation*}
\int_\la^{\bar{a}}\left[\lambda_1(a)+\lambda_2(a)A(a)\right]dR(a).
\end{equation*}
This is equivalent to maximizing $\int_\la^{\bar{a}}A(a)dS(a)$, where $dS(a)\equiv \lambda_2(a)dR(a)$.

\subsection{Moral Hazard and Educational Grades\label{sec:educ}} In the following extended application, we show how moral hazard can both justify monotonicity and create distorted priors.

\emph{The Setting.} A private university --- the sender --- designs a grading system to maximize tuition revenues. This is part of a larger problem in which it might choose an admissions cutoff, but we consider the ``second stage" where the school has already chosen its student body. Normalize admitted students to a unit measure, with abilities $a$ distributed according to $R$ on $[\la, {\bar a}]$. Upon entering, a student fully learns $a$, but before that she has only some prior, represented as a distribution over $[\la, {\bar a}]$. For instance, she may be close to a degenerate distribution on $[\la, {\bar a}]$ (full self-knowledge), or have the population belief $R$, or some third belief. Among these possibilities, we assume that there is  some continuous prior $F_0$, with $F_0(\la)=0$, which is the lowest according to first-order stochastic dominance among all priors. 

A market (our receiver) with shared prior $R$ infers student abilities from a \emph{learning level} or grade $\ell $ that the school chooses to make observable. Say a unit of ability is worth a dollar, and each unit of learning worth $\lambda $ dollars, and so the market pays
\begin{equation}
\mathbb{E} (a|\ell )+\lambda \ell 
\label{eq:gp}
\end{equation}
to a student with learning $\ell$, where the conditional expectation is determined by the prior $R$ as well as equilibrium strategies. Notice that $\ell $ has both intrinsic and signaling value.

\emph{Incentive-Compatible Learning.} The school chooses a  compact set of certified learning levels $L$. Learning nothing is always an option, so $0 \in L$. A student of ability $a$ chooses $\ell \in L$ by exerting effort at cost $c(a)\ell $, where $c'(a) < 0$. We assume the following condition:

\textbf{[C]} $c({\bar a}) > \lambda $, so that rewards to learning alone do not motivate any student.   

Every nonzero $\ell \in L$ will presumably be occupied by some ability type. A student \emph{could} choose $\ell \not \in L$, but the market observes only $\ell '= \max \{ \ell '' \in L | \ell'' \leqslant\ell \}$, so there is no point in doing so. In the spirit of direct mechanisms, suppose that the school ``assigns" learning $\ell (a) \in L$ to each ability type $a$. The value to an obeying student of ability $a$ is
\[
\mathbb{E}_R (a'|\ell (a)) +\lambda \ell(a) -c(a)\ell(a),
\]
where $\mathbb{E}_R (a'|\ell (a))$ is the expectation of ability $a'$ given that $\ell (a)$ is observed, and given that all students follow $\ell$. A standard single-crossing argument yields:
\begin{obs}
If $\ell (a)$ and $\ell (a')$ are optimal for $a$ and $a'$, with $a > a'$, then $\ell (a) \geq \ell (a')$.
\label{obs:sc}
\end{obs}
So incentive-compatibility restricts the school to a monotone categorization of abilities. $\ell$ could have separating intervals on which it is strictly increasing, and pooling intervals on which it is constant. These obviously correspond to a particular categorization $A_{\ell }$.
Incentive compatibility additionally implies that for every $a,a' \in [\la,\bar{a}]$,
\begin{equation}
A_{\ell}(a)+\lambda \ell (a) -c(a)\ell (a) \geq \max\left\{ A_{\ell}(a')+\lambda \ell(a') -c(a)\ell(a'), {\la} \right\}.
\label{eq:aic}
\end{equation}
The second constraint on the right hand side of (\ref{eq:aic}) is needed in case $\ell (\la) > 0$. Then the zero-learning choice is off-path, and suitable beliefs will be needed to guarantee incentive-compatibility. We presume that in such cases, the observation of 0 is associated with the belief that the student has the lowest ability $\la $.\footnote{This implements the best equilibrium from the perspective of a tuition-maximizing school.} The following observation tightly links incentive compatible learning functions $\ell$ to their corresponding categorizations $A_{\ell }$.
\begin{obs}
\textnormal{(i)} A learning function $\ell (a)$ is incentive-compatible if and only if  it is nondecreasing, if $\ell(\la)\in\left[0,\frac{A_\ell(\la)-\la}{c(\la)-\lambda}\right]$, where $A_{\ell}$ is the corresponding categorization, if $\ell $ is differentiable on any separating interval with 
\begin{equation}
\ell '(a) = \frac{1}{c(a) - \lambda},
\label{eq:ic1}
\end{equation}
and if at any threshold $t$ dividing two adjacent intervals,
\begin{equation}
\ell(t) = \up{\ell}(t) +\frac{A_{\ell}(t) - \up{A}_{\ell}(t)}{c(t) - \lambda },
\label{eq:ic2}
\end{equation}
where $\uparrow$ stands for left-hand limit.

 \textnormal{(ii)} Likewise, for every $A \in \mathcal{A}_R$ pick $\ell $ with $\ell(\la)\in\left[0,\frac{A(\la)-\la}{c(\la)-\lambda}\right]$, satisfying \textnormal{(\ref{eq:ic1})} and \textnormal{(\ref{eq:ic2})}. That describes all incentive-compatible  $\ell$ such that $A_{\ell} = A$.
\label{obs:ic}
\end{obs}

\emph{Tuition and School Payoffs.} The school sets a single  tuition level.  Type $F_0$ has the lowest willingness to pay, so  the school must maximize the expected payoff of this type, before fees.
That is, the school chooses  incentive compatible $\ell $ to maximize 
\begin{equation}
\int_{\la }^{\bar a}[ A_{\ell}(a) + \{ \lambda- \sigma c(a) \}  \ell (a) ]dF_0(a),
\label{eq:7}
\end{equation}
where $\s \in [0, 1]$ is the extent to which parents internalize effort costs at the ex-ante stage. The extent of this internalization will determine not just the level of the tuition (which is a relatively minor consideration, at least for the analysis), but also the school's ``attitude" towards the intrinsic value of learning; more on this immediately below. We now link the school problem to our more abstract setting, thereby permitting a full solution of it. 
 
\emph{Solution to the School Problem.} 
Consider two cases. If $\lambda >\s \int_{\la}^{\bar a}c(a)dF_0(a)$, then ``$F_0$-parents" value, on average, intrinsic learning relative to cost. It is obvious that no matter what categorization $A$ the school seeks to implement, its associated learning function must have the highest possible starting point. That is, recalling part (ii) of Observation \ref{obs:ic}, initial learning $\ell (\la)$ must be set equal to the upper bound  $\frac{A(\la)-\la}{c(\la)-\lambda}$ for any choice of $A \in A_{\mathcal{R}}$.

Otherwise, $\lambda \leqslant \s \int_{\la}^{\bar a}c(a)dF_0(a)$. Now learning isn't intrinsically valued by $F_0$-parents, so the school optimally sets  $\ell (\la) = 0$. That motivates the definition: for any $A \in A_{\mathcal{R}}$: 

\begin{align}
\ell^*(A) = \begin{cases}
\displaystyle{\frac{A (\la)-\la}{c(\la)-\lambda}} \quad & \text{ if  } \quad \displaystyle{\lambda >\s \int_{\la}^{\bar a}c(a)dF_0(a)} \\
\qquad 0 & \text{ if } \displaystyle{\lambda \leqslant \s \int_{\la}^{\bar a}c(a)dF_0(a)}
\end{cases}
\label{eq:defn}
\end{align}
\begin{proposition} Assume Condition C. For every $A \in {\mathcal A}_R$, pick unique $\ell$ as described in Observation \ref{obs:ic}, with $\ell(\la)=\ell^*(A)$ as defined in \textnormal{(\ref{eq:defn}). Then school payoff is given by}
\begin{equation}
\int_{\la}^{\bar{a}}\left[A(a)+\{ \lambda-\sigma c(a)\}\ell(a)\right] dF_0(a)=\int_{\la}^{\bar{a}}A(a)dS(a){+K}
\label{eq:obj}
\end{equation}
where {$K$ is a constant and}
\begin{align}
 \label{eq:G}
S(a) & =F_0(a) + \int_a^{\bar{a}}\frac{\sigma c(x)-\lambda}{c(a)-\lambda}dF_0(x)\text{ if }a\in(\la,\bar{a}], \text{ and}\\
S(\la) &= \min\left\{0, \int_\la^{\bar{a}}\frac{\sigma c(x)-\lambda}{c(a)-\lambda}dF_0(x)\right\}   .
\nonumber
\end{align}
The function $S$ is left-continuous and has bounded variation with at most one discontinuity. Also, $S(\la)$ is finite and $S(\bar a ) =1$. This, along with the fact that $A$ is right-continuous with at most countably many discontinuities guarantees that the Stieltjes integral in \textnormal{(\ref{eq:maxa})} is well-defined, and that all the assumptions of the baseline model are satisfied. 

The school problem is solved by finding a solution $A^*$ to the optimal categorization problem with $R$ as the receiver's distribution and $S$, defined in \textnormal{(\ref{eq:G})}, as the sender's distribution. The optimal learning function is the unique $\ell$ associated with $A^*$ with $\ell(\la)=\ell^*(A^*)$.
\label{pr:obj*}
\end{proposition}
Proposition \ref{pr:obj*} fully removes $\ell $ from the analysis, as well as its attendant moral-hazard implications, and converts a more complex model into our simpler categorization model. In so doing, it reveals three reasons for the ``induced prior" $S$ to be different from $R$. First, the sender may, in effect, be delegated to work on behalf of someone with a distinct prior. Here, this is the student or parent with lowest belief $F_0$. The school is pushed to it in the interests of tuition revenue. The second stems from ancillary incentive constraints that might be involved in revealing quality; here, these have to do with learning. Third, the actions taken to signal quality (school performance) may have direct payoff effects. All three enter equation (\ref{eq:G}). These considerations can additionally cause $S$ to depart from a cdf. 

At a somewhat more technical level, a continuous prior can even acquire a discontinuity as a result of incentive constraints. Remember that parental priors are continuous, with $F_0(\la) = 0$. In the first of the two cases where learning is intrinsically valued, $S(\la)$ as defined in (\ref{eq:G}) is negative --- it isn't a cdf any more --- but the entire induced function $S$ is easily seen to be continuous. If, on the other hand, learning is not intrinsically valued,  $S(\la)$ is set equal to $0$, generating a discontinuity in $S$, while $\lim_{a \downarrow \la} S(a) > 0$. In any case, our techniques apply to either case, but the latter will generate an initial pool of zero learning --- perhaps defining a ``party school."

\emph{Pooling.} An intriguing question which deserves a more detailed exploration is whether grade pooling is a pervasive property of an educational system.
Theorem \ref{th:1} and Proposition \ref{pr:pool} can be applied to shed some light on this issue. Proposition \ref{pr:fullp} below formalizes the following claims: a school will want to pool all ability types when (a)  the market places a low relative value on learning, as in Lizzeri (1999) where learning has no value;   (b) if students fully internalize their cost of learning ex-ante,  and (c) if the lowest belief student is certain ex-ante  that she is the lowest ability type.
\begin{proposition} A sufficient condition for full pooling to be optimal is:
\begin{equation}
\int_a^{\bar{a}}[\sigma c(x)-\lambda ]dF_0(x)\geqslant 0 \hskip10pt \text{ for every }a\in[\la,\bar{a}]
\label{eq:po}
\end{equation}

\textnormal{(i)} \textnormal{(\ref{eq:po})} is satisfied when $\lambda=0$. Moreover, if there is $\lambda>0$ such that it is satisfied, then it is also satisfied for any $\lambda'<\lambda$.

\textnormal{(ii)} \textnormal{(\ref{eq:po})} is satisfied when $\sigma=1$. Moreover, if there is $\sigma>0$ such that it is satisfied, then it is also satisfied for any $\sigma'>\sigma$.

\textnormal{(iii)} \textnormal{(\ref{eq:po})} is satisfied when $F_0$ is degenerate at $a=\la$. Moreover, if there is $F_0$ such that it is satisfied, then it is also satisfied for any $F_0'$ such that $F_0$ first order stochastically improves over $F_0'$.
\label{pr:fullp}
\end{proposition}

\begin{figure}[t!]
\includegraphics[width=.475\textwidth]{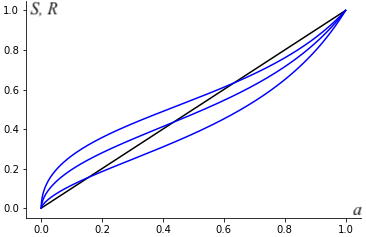}
\includegraphics[width=.475\textwidth]{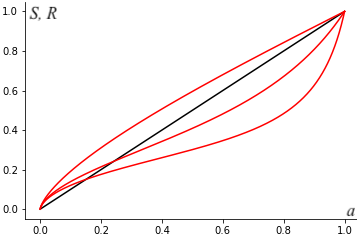}
\caption{\small{Schooling with lower censorship. In both panels, $R$ is uniform and in black. On the left, $S$ is shown in blue for three  values of $\gamma$: $S$ increases in fosd as $\gamma$ increases. On the right, $S$ is shown in red for three  values of $\lambda$: again, $S$ increases in fosd as $\lambda$ increases.}}
\label{fig:Gex}
\end{figure}

\emph{Schooling with Lower Censorship}.  Let $R$ be  uniform on $[0,1]$, and $F_0=a^\gamma$ for $\gamma\in[0,1]$. If $\gamma=1$, then $F_0=R$: students have no ex-ante information about their own abilities. For lower $\gamma$, $F_0$ is first-order stochastically dominated by $R$, so that the lowest types are pessimistic relative to the population average. The lower is $\gamma$, the further the belief of the most pessimistic type from that of the average agent. Set  $c(a)=1/a$ and $\lambda < 1$. Finally, set $\sigma=0$, which means that the cost of effort is not internalized at all by the  parents when choosing to join the school.

With these functional forms, use equation (\ref{eq:G}) to map the school's ``distorted prior" $S$: 
\begin{align}
S(a)=\frac{a^\gamma-\lambda a}{1-\lambda a}
\label{eq:exG}
\end{align}
In this special case, $S$ is a cdf (see Figure \ref{fig:Gex}). Since $S$ has a concave-convex shape,  Theorem \ref{th:1} immediately implies that lower censorship, whereby all abilities below a threshold are pooled together and all abilities thereafter are fully revealed, is the optimal categorization. 
\begin{proposition}
There is $\tilde{a}\in(0,1]$ such that the solution to the school problem is to pool agents with $a\in[0,\tilde{a})$, and fully reveal the ability of all agents with $a\geq\tilde{a}$.
\label{pr:ex}
\end{proposition}
We can easily compute $\tilde{a}$, and perform comparative statics with respect to $\lambda$ and $\gamma$. When the lowest-belief type is more optimistic, i.e. $\gamma$ is higher, $\tilde{a}$ is lower and there is more separation. The connection with the market value for learning is not monotonic. Initially, higher value for learning induces more separation, but for high values of $\lambda$, increasing $\lambda$ leads to more pooling.

If, otherwise, the distorted prior were convex-concave, upper censorship would be optimal. Other papers in the literature find that upper or lower censorship are optimal signals for the sender.  Alonso and Camara (2016b) and Kolotilin, Mylovanov and Zapechelnyuk (2019) show conditions for optimality of lower and upper censorship -- respectively, that the sender's payoff as a function of the receiver's posterior mean be concave-convex and convex-concave. In both those papers, they consider environments where the sender's payoff does not directly depend on the state and the sender's payoff is nonlinear in the receiver's posterior mean. In our model, the sender's payoff is state dependent and linear in the receiver's posterior mean. Perhaps surprisingly, these environments are distinct and cannot be mapped to each other.

\section{Remarks on Monotonicity\label{sec:mono}} The monotonicity restriction is central to our analysis. We have proposed a method to find an optimal categorization under that restriction.

Suppose instead that all possible categories can be used. One can use a version of Kamenica and Gentzkow's (2011) technique to directly create a concavification argument for the resulting problem.\footnote{ In fact, Alonso and Camara (2016) do adapt the Kamenica-Gentzkow concavification argument to the case of sender and receiver with heterogeneous priors.} Arbitrary categorization allows posteriors to be convexified to any  degree, so that {non-concavities} in the seller's value function (on the domain of {her} posteriors) can be removed. These convexifications, achieved through belief-splitting, are not available when categories are monotone, necessitating a different approach.

Another way to appreciate the distinction is to see how monotonicity converts a problem solvable using ``local methods" into one that cannot be so tackled. Consider a pooling category in the form of a single interval $[a,b]$. Suppose that the sender entertains an alternative, which is to insert a small separating category $[c,d)$ within this interval; $a < c \simeq d < b$. With no monotonicity restriction, this can be done by continuing to maintain $[a, c)\cup [d,b]$ as a pooling category, so the variation $[c,d)$ can be inserted continuously and with only local effects. With monotonicity, this insertion immediately splits the pooling category into two separate intervals, and the effects are large and global.

One possible reaction that any modeling of sender behavior ``should" freely permit all categorizations. After all, her payoff would generally be higher. We disagree. Monotonicity can be a real constraint; it isn't just a modeling device. There are at least two constraints that come to mind. The first is legal. For instance, if a credit rating agency grades the riskiness of debt issuers, a non-monotone categorization could invite a lawsuit from an issuer who can plausibly argue that they received a worse rating than other issuers with more risky debt. A similar legal constraint arises in questions of affirmative actions, where it is often argued that demonstrable meritocracy has to take precedence over considerations of race or gender. The second is incentive-based. As argued in some detail in Section \ref{sec:educ},  monotonicity could emerge as the outcome of a broader mechanism design problem which in which incentive constraints need to be respected.  

Additionally, there may be methodological reasons. For instance, certain problems translate into equivalent monotonic problems. Kolotilin and Zapechelnyuk (2019) establish an equivalence between balanced delegation problems and monotone persuasion problems. Our method finds the optimal signal in their linear persuasion environment when the sender's value is linear in the receiver's action. Within their equivalence class, we can therefore solve for the optimal delegation set.

And of course, monotonicity affects the results. In Section \ref{sec:exs}, we argued that our method can find the optimal monotonic signal in Rayo and Segal (2010). The resulting structure is quite different. For instance, Rayo and Segal characterize the optimal signal (without monotonicity restrictions), one of their main results is that two prospects $(v,\pi)$ and $(v',\pi')$ can only be optimally pooled if they are \emph{not ordered}: $(v-v')(\pi-\pi')<0$. This is no longer true when we look at the optimal monotonic categorization. It is possible for a pooling interval to contain a region where $\pi$ is increasing --- which is equivalent to $S\circ R^{-1}$ being convex --- and therefore pool two ordered prospects.

Furthermore, Rayo and Segal show (2010, Lemma 3) that in any optimal signal, all prospects pooled in the same signal realization must have payoffs $(v,\pi)$ that lie on a straight line with nonpositive slope. One implication is that if $\pi(v)$ does not have an interval where it is linear and downward sloping, then every realization of the optimal signal pools together at most two qualities.\footnote{Kolotilin and Wolitzky (2020) make the similar point that optimal signals are ``pairwise,'' in the sense that each induced posterior distribution has at most binary support.} Optimal monotonic signals, on the other hand, may pool together prospects containing more than two qualities. 

In summary, these are different structural features. They could coincide or continue to differ in special settings. For instance, it can be shown that the necessary and sufficient conditions for full separation are no different with or without monotonicity, but the same is not true of the conditions that characterize full pooling.\footnote{Full separation is optimal if and only if  $(v-v')(\pi(v)-\pi(v'))>0$ for every $v$ and $v'$. This is \emph{equivalent} to $S\circ R^{-1}$ being globally convex, the condition we found in Proposition 
\ref{pr:pool}. So the necessary and sufficient conditions for full separation are unaffected by the monotonicity constraint.
 But this is not true of full pooling. Without monotonicity, that condition (Rayo and Segal 2010) asks for $\pi(v)$ to be decreasing and affine. This is a  stronger condition than the full pooling condition we found in Proposition \ref{pr:sep}. In fact, this condition is stronger than  $S\circ R^{-1}$ being globally concave -- which is equivalent to $\pi(v)$ being decreasing. The conditions for full pooling and full separation found by Alonso and Camara (2016a), when sender payoff is linear in the receiver's action, are equivalent to the ones in Rayo and Segal (2010).} 

\section{Nonlinear Sender Payoffs\label{sec:nl}} Consider the following generalization of our model, in which the receiver takes an action  predicated on his conditional expectation (after learning the category membership of the object), and the sender obtains a payoff from that action. If quality $x$ is revealed, the sender's payoff is $U(x)$ and the receiver's payoff is $V(x)$. Otherwise, if $x$ is pooled on $[a,b)$, the receiver's action is based on her posterior $A(x)= \E _R(V(y) | [a,b))$, with resulting sender payoff  $U(A(x))$. A simple change of variables allows us to treat one of these payoff functions --- say $V$ --- as linear, which is what we do, but it is understood that any curvature restrictions on one of the functions (e.g., concavity) translate into the opposite restriction (e.g., convexity) on the other.

The sender picks categorization $A \in \mathcal{A}_R$ to 
\begin{equation*}
\mbox{maximize } \int_\la^{\bar{a}} U(A(x)) dS(x),
\label{eq:maxan}
\end{equation*}
where --- with $V$ normalized to be linear --- the receiver's payoff $A(x) = x$  in any separating region, and $A(x) = \E _R(y | [a,b))$ in any pooling region $[a,b)$ of $A$. 

This is a non-trivial extension. Generally, no  change of variables can convert it into our original maximand (\ref{eq:maxa}).
Curvature restrictions on $U$ allow us to make some progress, but recall from the normalization above that there is no particular presumption that $U$ be concave or convex. Below, we discuss some implications of curvature restrictions.
\begin{proposition}
Suppose that $U$ is convex. Then there is an optimal categorization $A^*$ such that any separating quality $x$ in any optimal solution to the categorization problem with linear payoff $U(x)=x $ is also separating under $A^*$. 
\label{pr:sup}
\end{proposition}
Proposition \ref{pr:sup} makes a seemingly intuitive point: the convexity of $U$ is a force that pushes the sender towards revealing additional information. Therefore the separating regions expand compared to the benchmark of the linear case. Of course, if $U$ is a concave function, we would then expect the opposite effect to be in place, so that the sender wants to pool to a greater degree. One reason why we refer to the proposition as  ``seemingly intuitive" is that the latter result does not hold in general. 


\begin{figure}[t!]
\includegraphics[width=1\textwidth]{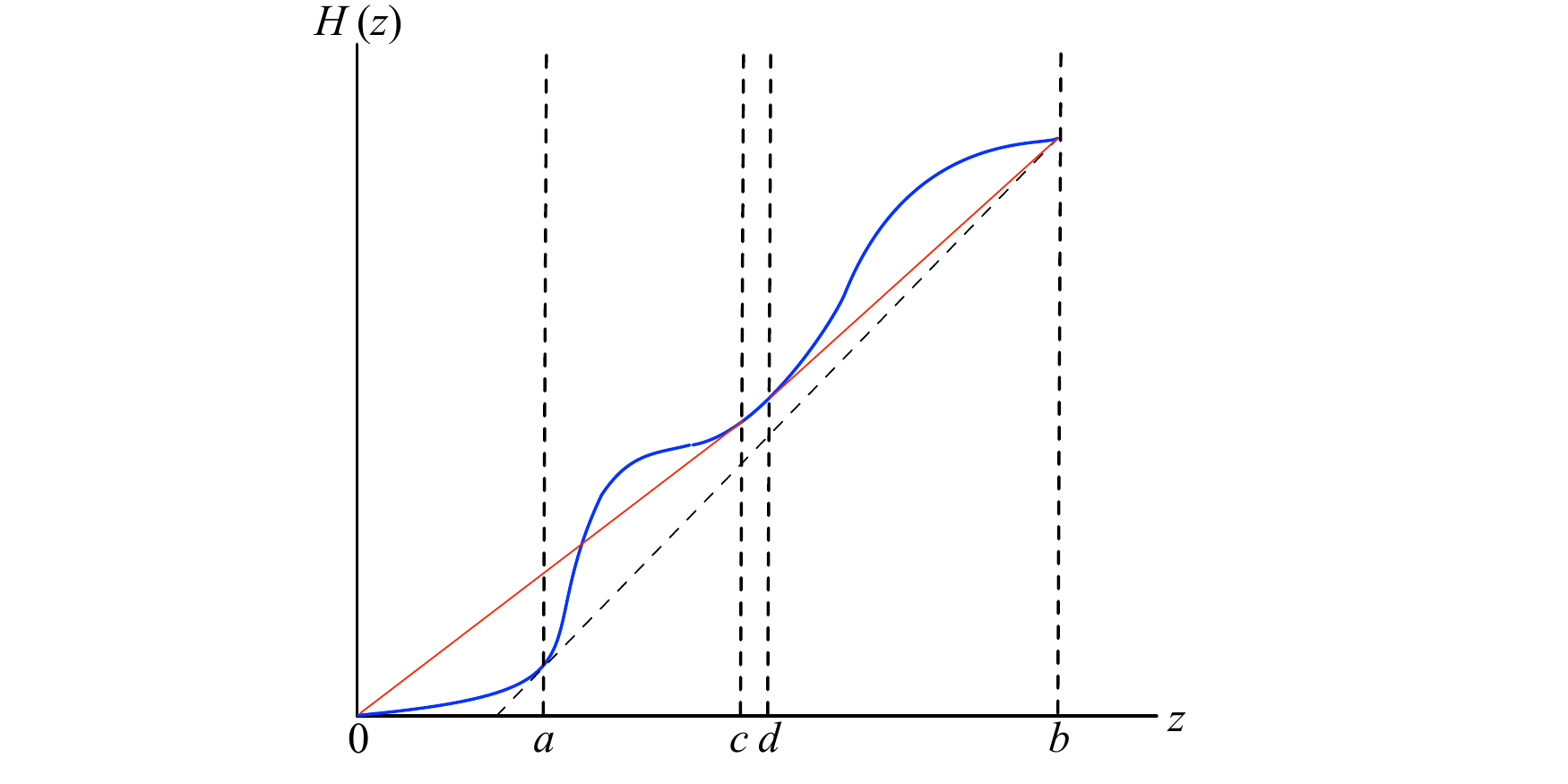}
\caption{\small{An example in which pooling regions do not expand with the concavity of $U$. The curved line is $H = S\circ R^{-1}$.  The dashed line connecting $(a,H(a))$ and $(b,H(b))$ indicates the pooling region in our linear model. The red lines connecting $(0,H(0))$ to $(c,H(c))$ and $(d,H(d))$ to $(b,H(b))$ indicate the pooling regions in the nonlinear case.}}
\label{fig:concu}
\end{figure}

Consider the example described in Figure \ref{fig:concu}. In the linear case, the optimal categorization is given by separating every quality on the interval $[0,a)$ and pooling every quality on the interval $[a,b]$. Suppose, now, that $U$ is concave --- in fact, highly concave up to some value $c$ and close to linear thereafter. The extreme initial concavity can cause pooling on the interval $[0,c)$.  Because $U$ is almost linear after $c$, we can solve for the optimal categorization thereafter (between $c$ and $b$) using our method. That will pick out $[c,d)$ as some initial separating interval and $[d,b]$ as a second and final pooling interval. In summary, the pooling regions with this sender payoff function are $[0,c)$ and $[d,b]$, and their union is not a superset of the pooling region for the linear case.

As a final remark, we note that nonlinear sender payoffs can sometimes be directly embedded into our linear framework, though the extent to which our techniques extend is an interesting open question at this stage. Suppose sender and receiver share a common prior ($S=R$) and the sender's payoff function is given by $U(x,A)=\lambda_1(x)A+\lambda_2 A^2$, for  $\lambda_1:[\la,\bar{a}]\rightarrow \mathbb{R}$ and $\lambda_2\in\mathbb{R}$.\footnote{Here we are restricting to $S=R$ and allowing $U$ to be state-dependent, as well as nonlinear. If $\lambda_2=0$, then, as in the discussion in Section \ref{sec:app},  we return to the linear case with distinct priors, where $dS(x)=\lambda_1(x)dR(x)$.}
In this case, the sender picks $A\in\mathcal{A}_R$ to maximize
\[
\int_\la^{\bar{a}}\left[\lambda_1(x)A(x) +\lambda_2 A(x)^2\right]dS(x).
\]
By using $S=R$ and the definition of $A(x)$, it is easy to see that
\begin{align*}
\int_\la^{\bar{a}}\left[\lambda_1(x)A(x) +\lambda_2 A(x)^2\right]dS(x) & = 
\int_\la^{\bar{a}}\left[\lambda_1(x)+\lambda_2 A(x)\right]A(x)dR(x)\\
& =\int_\la^{\bar{a}}\left[\lambda_1(x)+\lambda_2 x\right]A(x)dR(x)
=\int_\la^{\bar{a}}A(x)d\hat{S}(x),
\end{align*}
where $d\hat{S}(x)=\left[\lambda_1(x)+\lambda_2 x\right]dR(x)$. Now our method can be applied to find the optimal monotonic signal. Notice how the nonlinearity of $U$  induces a different sender prior, so this is yet another argument that one could add to the discussion in Section \ref{sec:exs}.\footnote{Notice also that we could have also combined this idea with state-dependent preferences --- for instance, $\lambda_1$ could have been taken to be a function of $x$ without making any difference to the analysis.} 

When $\lambda_2$ increases, $d\hat{S}$ becomes ``more increasing'', or equivalently, $\hat{S}\circ R^{-1}$ becomes ``more convex." This means that the sender has stronger incentives to create separating intervals. If $\lambda_2$ is large enough, $\hat{S}\circ R^{-1}$ is in fact convex, in which case the optimal monotonic categorization is full separation. On the other hand, if $\lambda_2$ is small enough (perhaps negative), then the optimal monotonic categorization is full pooling. 

\section{Bibliographical Notes\label{sec:lit}}
We contribute to the literature on ``information design," stemming from Rayo and Segal (2010) and Kamenica and Gentzkow (2011). The problem in our paper has three main features: (i) the sender's payoff is linear in the receiver's posterior mean about the state; (ii) sender and receiver hold distinct priors about the distribution of states; and  (iii) the sender must choose monotonic categorizations. We provide a method that generates the optimal monotonic categorization for any pair of priors held by the sender and the receiver.

Our paper is related to Alonso and Camara (2016a), who extend the concavification method to the case of heterogeneous priors between sender and receiver, but do not impose monotonicity. As has been widely noted in the literature, the concavification method alone is typically not sufficient to characterize the optimal signal, especially when the state and action space is large as in our model. We add to the literature by providing a simple method to find the optimal monotonic signaling structure. In Section \ref{sec:mono}, we compare our solution to optimal signals not subject to the monotonicity restriction. 

Galperti (2019) and Kartik, Lee and Suen (2020) also study environments with heterogeneous priors across agents. Both these papers study problems different from ours, but some results in Kartik, Lee and Suen (2020) have parallels for the  specific question of when full pooling or separation of qualities is optimal.

Kolotilin (2018) and Dworczak and Martini (2019)  study optimal signaling structures in environments where the receiver's action depends only on their posterior mean. Unlike us, they focus on cases where the sender's payoff depends on the state only through the receiver's action, an assumption that Lipnowski and Ravid (2020) call ``transparent motives."\footnote{In Kolotinin (2018), some results apply to a broader class of persuasion problems. Assumption 2 implies that the receiver's action depends only on their posterior mean and that the sender's payoff depends on the state only through the receiver's action.}\footnote{Lipnowski and Ravid (2020) study a cheap talk problem, rather than Bayesian persuasion.} In our paper, sender's payoffs depend linearly on the receiver's action, but also depend on the state via their prior over states, which is distinct from the prior held by the receiver. In addition, the sender is restricted to choose a monotonic information structure, while Kolotilin (2018) and Dworczak and Martini (2019) deal with the optimality of a structure in the class of \emph{all} signaling structures, monotone or not. We note that one of the results in Dworczak and Martini (2019) provides conditions for the optimality of monotone structures, though again for the case of a ``transparent motive.'' Kolotilin (2018) also provides conditions for a special type of monotone structure, which he calls interval revelation, to be optimal. Kolotilin and Li (2019) provide some characterization results for monotone persuasion, again in a setup where the sender's payoff is state independent.

As we show in Section \ref{sec:app}, there is an equivalence between heterogeneous sender-receiver priors and state dependent priors for the sender. Given that equivalence, the problems studied in Rayo and Segal (2010) and Rayo (2013) are equivalent to problems where sender and receiver have distinct priors. Rayo and Segal (2010) characterize optimal signaling structures without the monotonic restriction. We use their characterization in Section \ref{sec:mono} to compare the monotonic to the non-monotonic solution. In passing, we note an equivalence between Rayo and Segal (2010) and Alonso and Camara (2016a) when the sender's payoff is linear in the receiver's posterior.

Rayo (2013) characterizes optimal monotonic signals. Let $h$ be a function describing how the sender's payoff depends on the state. Rayo shows that there exists a unique set of maximal intervals such that within each interval, the covariance between $h$ and any non-decreasing function of the state defined on that same interval is negative. He argues that, if a solution to the sender's problem exists, it must pool together the states in each of the maximal intervals and reveal all the other states. We show that a solution to the sender problem can be found by chalking out the lower convex envelope to a particular function of sender and receiver priors. Intervals where this envelope departs from the function constitute optimal pooling intervals, and all other intervals are separating. In the Supplementary Appendix, we show that under mild assumptions, we can write down an algorithmic procedure that generates the pooling and separating intervals in our solution, which leads to a more transparent interpretation of  the characterization in Rayo (2013).

In part, our analysis can be viewed as by bridging --- and going beyond --- Rayo's specific problem to a more abstract setting where sender and receiver have distinct priors. This allows us to interpret the sender's incentives to pool states within certain intervals as due to ``greater local pessimism'' than the receiver, though the converse (as already mentioned in the Introduction) is not true. In particular, the interpretation of state-dependent payoffs as heterogeneous priors motivates the exercise we discuss in Section \ref{sec:pre}. We show that if we flip the priors of the sender and receiver, the optimal monotonic signal  does not simply get ``flipped'' as well. Rather, there is a preponderance of pooling: some states are pooled both in the original and in the ``flipped'' problem.


In our digression into moral hazard in the grading problem, we are related to Boleslavsky and Kim (2020), who also look at an environment in which signaling structures motivate costly effort. They extend the concavification method in Kamenica and Gentzkow (2011) to this environment. In our application, agents exert effort after drawing their ability and the incentive constraint implies that signaling structures must be monotone. This is not the case in Boleslavsky and Kim (2020), where costly effort improves the distribution of types that will be drawn from. As already mentioned, it is not always easy to characterize the optimal signaling structure using the concavification method when the state and action spaces are large. In fact, most of the characterizations provided in Boleslavsky and Kim (2020) relate to the case of a binary state space or binary actions. Our procedure allows us to characterize the optimal structure that incentivizes agents and signals their abilities to the receivers in the case of a continuum of agent types and a continuum of receiver actions. Rodina and Farragout (2016) and Saeedi and Shourideh (2020) also study the problem of a principal who wants to improve an agent's investment in productivity when the only instrument at hand is an information disclosure policy. Their environment is different from ours both in how the agent's effort decision is set up and in the grading schemes that the authors allow for. 

Moldovanu, Sela and Shi (2007) study the problem of a principal who wants to incentivize agents who are playing a contest for status to exert effort by designing ``status categories''. In the context of a contest for status, the principal is only designing how coarsely to order agents from best to worst, rather than concealing or revealing their actual qualities. They find that coarser schemes, such as putting agents in two categories, are optimal when the distribution of qualities is sufficiently concave. Dubey and Geanakoplos (2010) also study the design of categories to incentivize agents in games of status, allowing for both cardinal and ordinal comparisons. They too find that often agents are best motivated by coarse categories. Our abstract setting is free of moral hazard, though our educational application studies both adverse selection and moral hazard. Motivation by the use of a mix of coarse and revealing categories is discussed in the Supplementary Appendix. 

Doval and Skreta (2018) extending the concavification method to the case where the sender is subject to a certain class of incentive constraints and review a literature where their tools apply. Lipnowski, Mathevet and Wei (2020) study a problem where the sender is constrained to the receiver's decisions to pay attention to the signals he is sent. In their model, the sender has the same objective as the receiver, but does not take into account attention costs incurred by the receiver. Similarly, in our application, the school does not fully incorporate the student's effort costs of learning.  Guo and Shmaya (2019) study the problem of a sender who sends a message to a receiver with private information about his type in order to motivate him to take one of two actions. They show that an incentive compatible mechanism takes the form of nested intervals: each receiver type accepts on an interval of states and a more optimistic type's interval contains a less optimistic type's interval. They then develop tools to solve for the optimal signaling structure within this category of structures with nested intervals. In our paper, incentive compatibility similarly implies a restriction on the class of signaling structures available --- in our case, to monotone structures. 



\section{Proofs\label{sec:proofs}}
\emph{Proof of Observation \ref{obs:v}.} Part (i). Let Pool denote the collection of pooling intervals with generic element $[p,p')$.
Because $A\in\mathcal{A}_R$ and $R$ is strictly increasing, we have:
\begin{align*}
\int_\la^{\bar{a}} A(a) dS(a) &=\int_{[\la,\bar{a}]\setminus \text{Pool}}adS(a)+\sum_{\text{Pool}} \mathbb{E}_R\left(a|a\in [p, p')\right)\left[ S(p')-S(p)\right]\nonumber\\
&=\int_{[\la,\bar{a}]\setminus \text{Pool}}adS(a)+\sum_{\text{Pool}} \int_{p}^{p'}adR(a)\left[\frac{S(p')-S(p)}{R(p')-R(p)}\right]\nonumber\\
&=\int_{[\la,\bar{a}]\setminus \text{Pool}}ad\Psi (a, A)+\sum_{\text{Pool}} \int_{p}^{p'}ad\Psi (a, A)= \int_\la^{\bar{a}} ad\Psi (a, A),
\end{align*}
where in the penultimate step $d\Psi$ is well defined as $\Psi$ has bounded variation, and the last equality follows from the continuity of $\Psi$ in $a$.

For any categorization $A$ and $z \in [0, 1]$, consider the ``percentile weighting function":
\[
\Phi (z, A) \equiv \Psi (R^{-1}(z), A).
\]
Define  a \emph{percentile pooling interval} of $A$ as any interval $[w, w')$ such that $[R^{-1}(w), R^{-1}(w'))$ is a pooling interval of $A$.  Then
\begin{equation}
\scalebox{0.83}{$\displaystyle{\Phi (z, A )}$}=\begin{cases}\scalebox{0.83}{$\displaystyle{H(w)+(z-w)\left[\frac{H(w')-H(w)}{w'-w}\right]\text{  if }z \mbox{ is in some percentile pooling $[w, w')$}}$};\\[0.5em]
\scalebox{0.9}{$\displaystyle{H(z) \text{  otherwise.}}$}\end{cases}
\label{eq:phi}
\end{equation}
In particular, the percentile weighting associated with a categorization $A\in\mathcal{A}_R$ equals $H$ in percentile separating regions and is a straight line connecting $(w,H(w))$ and $(w',H(w'))$ in percentile pooling regions of the form $[w,w')$. This means that $\mbox{Graph}(\Phi(\cdot,A))\subset \mbox{Co}(\mbox{Graph}(H))$, which immediately implies $\Phi(z,A)\geqslant \breve{H}(z)$.
\qed

\emph{Proof of Proposition \ref{pr:pool}.}
When $A$ is the categorization that pools every quality, then $\Phi(z,A)=H(0)+z\left(H(1)-H(0)\right)$. By Theorem 1, full pooling is then a solution to the sender's problem if and only if $\breve{H}(z)=H(0)+z\left(H(1)-H(0)\right)$ for all $z\in (0,1]$. Now notice that this condition is equivalent to 
\begin{equation*}
\frac{H(z)-H(0)}{z}\geqslant H(1)-H(0)
\end{equation*}  
for all $z\in (0,1]$.
Using the definition of $H$, this condition can be rewritten as $S(x)-S(\la)\geq (1-S(\la))R(x)$ for all $x\in[\la,\bar{a}]$.
\qed

\emph{Proof of Proposition \ref{prop:poolint}.}
If $R$ fosd $S$ on the interval $[a,b]$, then for all $x\in [a,b]$, 
\begin{equation*}
\frac{S(x)-S(a)}{R(x)-R(a)}\geqslant \frac{S(b)-S(a)}{R(b)-R(a)}
\end{equation*}
Or, equivalently, for every $z\in[w, w']$, where $w= R(a)$ and $w' = R(b)$,
\begin{equation*}
H(z)\geqslant H(w)+(z-w)\left[H(w')-H(w)\right].
\end{equation*}
Because the straight line connecting $(w, H(w))$ and $(w',  H(w'))$ lies in $\mbox{Co}(\mbox{Graph}(H))$,  it must be that for $z\in[w,w')$, $\breve{H}(z)\leqslant H(w)+(z-w)\left[H(w')-H(w)\right]\leqslant   H(z)$. But then there are percentiles $z'$ and $z''$ with $z\leqslant w<w'\leqslant z'$ such that for $z\in[w,w')$, $H(z)$ belongs to the straight line connecting $(z',  H(z'))$ and $(z'',  H(z''))$. So there exists an optimal categorization that pools the interval of percentiles $[z',z'')$, which contains the interval of percentiles $[w,w')$. Equivalently, such categorization pools the interval of qualities $[R^{-1}(z'),R^{-1}(z''))$, which contains the interval $[a,b)$.

Now let's prove the second statement. Suppose $[a,b)$ belong to a pooling interval $[a',b')$ with $a'\leqslant a$ and $b'\geqslant b$. Also suppose $R$ does not fsod $S$ on $[a',b']$. Then there exists $z\in(R(a'),R(b'))$ such that 
$$H(z)< H(w)+(z-w)\left[H(w')-H(w)\right]$$
where $w=R(a')$ and $w'=R(b')$. But that means that $H(w)+(z-w)\left[H(w')-H(w)\right]\neq \breve{H}(z)$, and so $[a',b')$ cannot be a pooling interval in the optimal categorization.\qed

\emph{Proof of Proposition \ref{pr:sep}.} 
The first statement is immediate given Theorem \ref{th:1}. As for the second, $H$ is convex in $[0,1]$ if and only if for every $w,x,z\in [0,1]$ with $w< x\leqslant z$: 
\begin{equation*}
\frac{H(x)-H(w)}{x-w}\leqslant \frac{H(z)-H(w)}{z-w}
\label{eq:cvx}
\end{equation*}
Letting $a=R^{-1}(w)$ and $b=R^{-1}(z)$ and $y=R^{-1}(x)$, this condition is equivalent to: for all $a,b,y\in[\la,\bar{a}]$ with $a<y\leqslant b$,
\begin{equation*}
\frac{S(y)-S(a)}{R(y)-R(a)}\leqslant \frac{S(b)-S(a)}{R(b)-R(a)}
\label{eq:cvx2}
\end{equation*}
If $S$ is a strictly increasing cdf, then this condition is equivalent to $S\left(\cdot|(a,b)\right)$ first-order stochastically dominating $R\left(\cdot|(a,b)\right)$ for every $a,b\in[\la,\bar{a}]$, which is in turn equivalent to $S$ dominating $R$ in the likelihood ratio order. \qed

\emph{Proof of Proposition \ref{prop:sepint}.} Let $w=R(a)$ and $w'=R(b)$.
If $[a,b)$ is a subset of some separating interval of qualities of categorization $A$, then for all $x\in[w,w')$, $\Phi(x,A)=H(x)$. And if $H(x)\neq\breve{H}(x)$ for some $x\in[w,w')$, then by Theorem \ref{th:1}, $A$ is not an optimal categorization. 

Now notice that, if $H=\breve{H}$ on this interval, then $H$ is convex on it, and so $S$ dominates $R$ in the likelihood ratio order over this interval. \qed

\emph{Proof of Proposition \ref{pr:alt}.} See Supplementary Appendix. \qed

\emph{Proof of Proposition \ref{pr:flip}.} Suppose that an optimal scheme in the original problem has at least one separating interval; consider any such interval $[a_1, a_2)$.\footnote{As always, use the convention that the right bracket is closed if $a_2 = \ha$.} Let $[z_1, z_2)$ be the associated percentile intervals; that is, $z_1 = R(a_1)$ and $z_2 = R(a_2)$. By Proposition \ref{prop:sepint}, $H(z) = \breve{H}(z)$ on $[z_1, z_2]$ and is therefore locally convex. It follows that $H^{-1}(y)$ is locally concave on the percentile interval  $[y_1, y_2]$, where $y_1 = H(z_1)$ and $y_2 = H(z_2)$. 

Now consider the flipped problem $(S', R')$ with $S' = R$ and $R'=S$.  Then it is obvious that $H^{-1} = S' \circ R'^{-1}$. Applying the second part of Proposition \ref{prop:poolint} to the fact that $H^{-1}(y)$ is locally concave on the percentile interval  $[y_1, y_2]$, we see that $[y_1, y_2]$ can be mapped via $S$ to an interval of optimally pooled qualities in the flipped problem. But it is easy to verify that $y_1 = S(a_1)$ and $y_2 = S(a_2)$.\footnote{For each $i = 1, 2$, $y_i = H (z_i) = S \circ R^{-1}(z_i) = S(a_i)$.} Therefore every separating interval $[a_1, a_2)$ becomes a subset of some optimal pooling interval in the flipped problem, and so the first assertion of the proposition is established.

Now assume the conditions under which the second assertion is made. Pick any optimal separating interval $[a_1, a_2)$, say for the original problem.  Let $[z_1, z_2)$ be the associated percentile interval; that is, $z_1 = R(a_1)$ and $z_2 = R(a_2)$. By Proposition \ref{prop:sepint} again, $H(z) = \breve{H}(z)$ on this percentile interval and is therefore convex on this interval. It is nowhere locally affine, by assumption, so in fact $H$ is \emph{strictly} convex on $[z_1, z_2)$.  It follows that $H^{-1}(y)$ is strictly concave on the percentile interval  $[y_1, y_2]$, where $y_1 = H(z_1) = S(a_1)$ and $y_2 = H(z_2) = S(a_2)$.  Consequently, $H^{-1}(y)$ must depart from its lower convex envelope $\breve{H^{-1}}(y)$ entirely on the \emph{open interval} $(y_1, y_2)$.

\begin{figure}[t!]
\begin{center}
\hspace*{-0.25in} 
\subfloat[{\small $y'_i \neq y_i$ for both $i$.}]{
\includegraphics[width=0.43\textwidth]{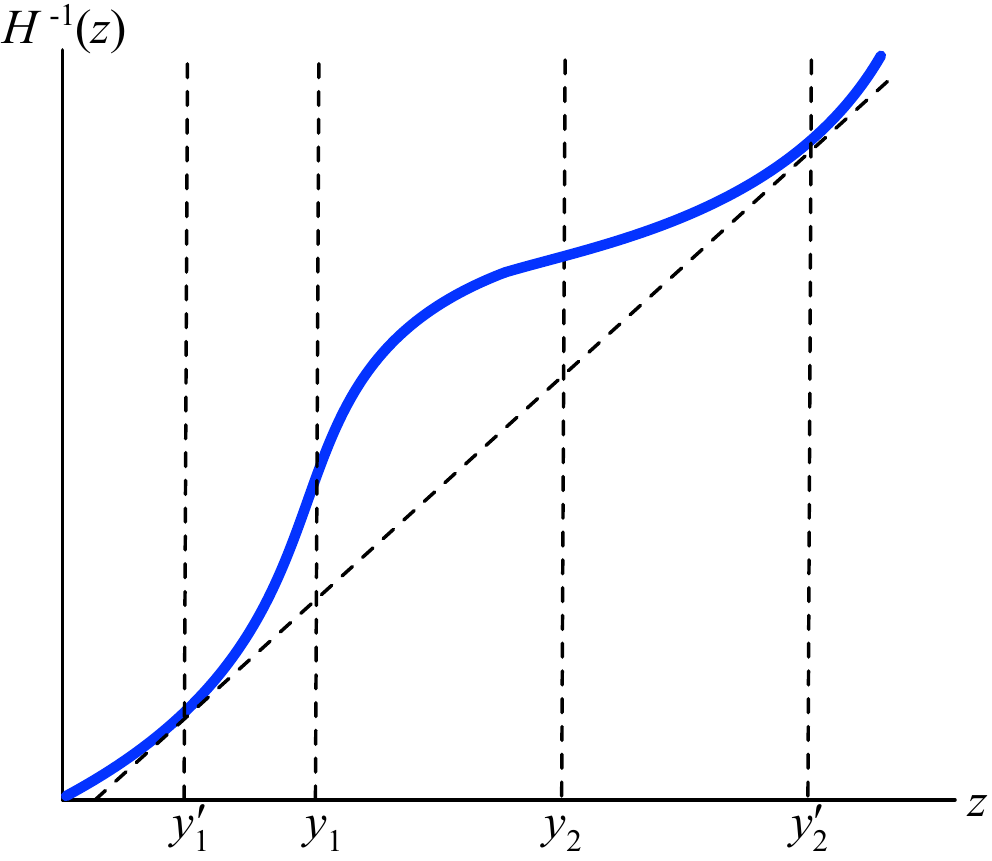}} 
\subfloat[{\small$y'_1 \neq y_1$, $y'_2 = y'_2$.}]{
\includegraphics[width=0.43\textwidth]{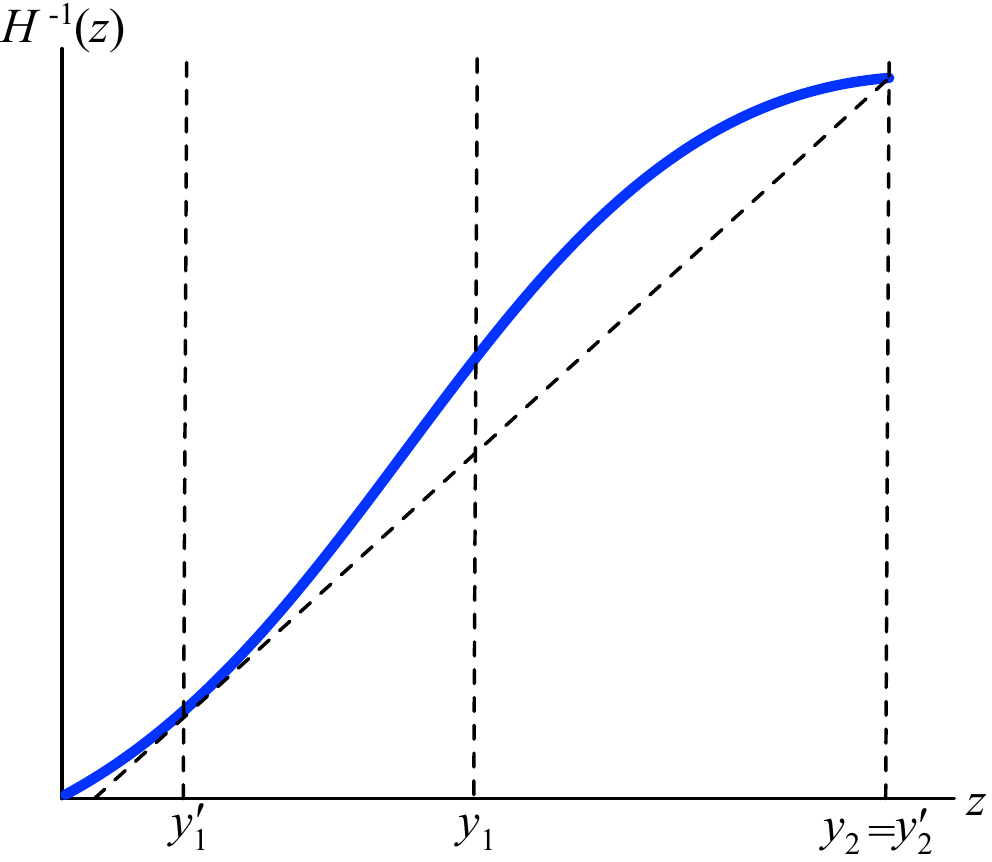}}
\end{center}
\begin{caption}
{\small{Illustration for proof of Proposition \ref{pr:flip}. $H^{-1}(z) \neq \breve{H^{-1}}(z)$ on a \emph{strict superinterval} $(y_1', y_2')$ of $(y_1, y_2)$, where $y'_i \neq y_i$ for at least one $i = 1,2 $.}\label{fig:Pooldom}}
\end{caption}
\end{figure}

Next,  because $H$ is not globally convex,  it follows from Proposition \ref{pr:sep} that either $a_1 > \la $ or $a_2 <\ha $ --- that is, the full interval of qualities cannot be separating.\footnote{For if it were, then by Proposition \ref{pr:sep}, $H = \breve{H}$, but the latter is convex by definition.} Because $y_i= S(a_i)$ for $i=1,2$, we have either $y_1 > 0$ or $y_2 < 1$, and  because  $H^{-1}$ is differentiable on the \emph{full} domain $(0,1)$,\footnote{This follows from the positive differentiability of $S$ and $R$ and the fact that $R$ maps \emph{onto} $[0,1]$.}  it follows that $H^{-1}(z) \neq \breve{H^{-1}}(z)$ on a \emph{strict superinterval} $(y_1', y_2')$ of $(y_1, y_2)$, where $y'_i \neq y_i$ for at least one $i = 1,2 $. Figure \ref{fig:Pooldom} provides a diagrammatic proof. By the second part of Proposition  \ref{prop:poolint}, all of $[y_1', y_2')$ \emph{must} be pooled, which proves the second assertion of the proposition. \qed

\emph{Proof of Observation \ref{obs:ic}.} Let $\ell(a)$ be an incentive-compatible learning function. {First, note that for an initial choice of $\ell(\la)\in\left[0,\frac{A_\ell(\la)-\la}{c(\la)-\lambda}\right]$, (\ref{eq:aic}) is satisfied for the lowest ability student.} By Observation \ref{obs:sc}, $\ell$ is nondecreasing. Next, take $a$ and $a'$ in the same separating interval. From  (\ref{eq:A}) and (\ref{eq:aic}), we have:
\[
\frac{1}{c(a')-\lambda}\geq\frac{\ell(a)-\ell(a')}{a-a'}\geq \frac{1}{c(a)-\lambda}.
\]
Take $a'\rightarrow a$ to obtain (\ref{eq:ic1}) on any separating interval. 
Now take $a\in[a_{k-1},a_k)$ for some $k>1$. Use (\ref{eq:aic}) to see that
\[
\ell(a)+\frac{A_{\ell}(a_k)-A_{\ell}(a)}{c(a)-\lambda}\leq\ell(a_k)\leq \ell(a)+\frac{A_{\ell}(a_k)-A_{\ell}(a)}{c(a_k)-\lambda},
\]
and send $a\rightarrow a_k$ to get (\ref{eq:ic2}). With $\ell(a)$ constant on pooling intervals, we conclude that (i) and (ii) are also sufficient for an incentive-compatible learning function $\ell(a)$. 

Conversely, let $A\in {\mathcal A}_R$. Define $\ell(a)$ with $\ell (\la) = 0$, with (\ref{eq:ic1}) holding on separating intervals of $A$, with $\ell $ constant on pooling intervals of $A$, and satisfying (\ref{eq:ic2}) --- with $A_{\ell} = A$ --- at every left edge $a_k$ of every interval. Standard arguments for differential equations ensure that $\ell $ is well-defined and unique. By Observation \ref{obs:ic}, $\ell $ is incentive compatible. \qed


\emph{Proof of Proposition \ref{pr:obj*}}. The proof will rely on the following lemmas:

\begin{lemma} Let $A \in \mathcal A$ and $\ell $ be the unique associated learning function as given by Observation \ref{obs:ic}. 
Let $dA$ and $d\ell$ be the Stieltjes measures associated with $A$ and $\ell$ respectively. Then $d\ell$ is absolutely continuous with respect to $dA$ and 
\[
\frac{d\ell}{dA}(a)=\frac{1}{c(a)-\lambda}
\]
is the Radon-Nikodym derivative of $d\ell$ with respect to $dA$.
\label{radon}
\end{lemma}
\begin{proof}
For any set $B\subset \left[\la,\bar{a}\right]$, $dA(B)=0$ trivially implies $d\ell(B)=0$ since $A$ and $\ell$ are constant or strictly increasing in exactly the same intervals. Hence $d\ell$ is absolutely continuous with respect to $dA$, and so there exists a Radon-Nikodym derivative between the two measures.
Now let $[b,b')\subset [\la,\bar{a}]$. Then $[b,b')$ is made up of countably many intervals $[c,c')\in \mathcal{C}$ on which both $A$ and $\ell$ are continuous and differentiable, along with at most countably many points of discontinuity $d\in\mathcal{D}$. It follows that
\begin{align}
&d\ell[b,b')=\ell(c)-\ell(b)=\sum_{(c,c')\in \mathcal{C}} \int_c^{c'}\ell'(x)dx+\sum_{d\in\mathcal{D}}\left[ \ell(d)-\ell^{\uparrow}(d)\right]\nonumber\\
&=\sum_{(c,c')\in \mathcal{C}} \int_c^{c'}\frac{1}{c(x)-\lambda}A'(x)dx+\sum_{d\in\mathcal{D}}\frac{1}{c(d)-\lambda}\left(A(d)-A^{\uparrow}(d)\right)=\int_b^{b'}\frac{1}{c(x)-\lambda}dA(x)\nonumber
\end{align}
where the third equality uses Observation \ref{obs:ic}.

Because the intervals $[b,b')\in[\la,\bar{a}]$ generate the Borel $\sigma$-algebra in $[\la,\bar{a}]$, we  conclude that $\frac{d\ell}{dA}(x)=\frac{1}{c(x)-\lambda}$ is the Radon-Nikodym derivative of $d\ell$ with respect to $dA$.
\end{proof}

\begin{lemma}\textbf{Integration by Parts.} If $P$ is an $Q$-integrable function on $[\la,\bar{a}]$, then $Q$ is $P$-integrable on $[\la,\bar{a}]$ and
\begin{align}
\int_{\la}^{\bar{a}}P(x) dQ(x)=P(\la)\int_{\la}^{\bar{a}}dQ(x)+\int_{\la}^{\bar{a}}\int_{\la}^xdQ(y)dP(x) \nonumber
\end{align}
\label{parts}
\end{lemma}
\begin{proof}
If P is Q-integrable, then the standard integral by parts formula yields
\begin{align}
\int_{\la}^{\bar{a}}P(x) dQ(x)=P(\bar{a})Q(\bar{a})-P(\la)Q(\la)-\int_{\la}^{\bar{a}}Q(x) dP(x) \label{eq:parts}
\end{align}
Rearrange (\ref{eq:parts}) to get: 
\begin{align}
\int_{\la}^{\bar{a}}P(x) dQ(x)&=\left[P(\la)+\int_{\la}^{\bar{a}}dP(x)\right]Q(\bar{a})-P(\la)Q(\la)-\int_{\la}^{\bar{a}}Q(x) dP(x)\nonumber\\
&=P(\la)\int_{\la}^{\bar{a}}dQ(x)+\int_{\la}^{\bar{a}}\left(\int_{\la}^{\bar{a}}dQ(y(-\int_{\la}^xdQ(y)\right)dP(x)\nonumber\\
&=P(\bar{a})Q(\bar{a})-P(\la)Q(\la)-\int_{\la}^{\bar{a}}Q(x) dP(x)\nonumber
\end{align}
\end{proof}
\textbf{Case 1. }First assume $\lambda >\s \int_{\la}^{\bar a}c(a)dF_0(a)$, so that $\ell^*(\la)=\frac{A(\la){-\la}}{c(\la)-\lambda}$. 

Set $P=A$ and $Q=F_0$ in Lemma \ref{parts}. Because $F_0$ is continuous and of bounded variation, the relevant integral is defined, and
\begin{equation}
\int_{\la}^{\bar{a}}A(a)dF_0(a) = A(\la)\int_{\la}^{\bar{a}}dF_0(a)+\int_{\la}^{\bar{a}}\int_a^{\bar{a}}dF_0(x)dA(a).
\label{apartsa}
\end{equation}
Next, setting $P=\ell $, and $dQ(x) = [\lambda - \sigma c(x)] dF_0(x)$ in Lemma \ref{parts}, and noting again that $Q$ is continuous and of bounded variation, we see that
\begin{equation}
\scalebox{0.82}{$\displaystyle{\int_{\la}^{\bar{a}}[\lambda-\sigma c(a)]\ell(a)dF_0(a) = \ell(\la)\int_{\la}^{\bar{a}}[\lambda-\sigma c(a)]dF_0(a)+\int_{\la}^{\bar{a}}\int_a^{\bar{a}}[\lambda-\sigma c(x)]dF_0(x)d\ell (a)}$}.
\label{apartsb}
\end{equation}
Recall that $\ell(\la)=A(\la)/[c(\la)-\lambda]$. Use this in (\ref{apartsb}),  and invoke  Lemma \ref{radon} to get: 
\begin{align}
\scalebox{0.82}{$\displaystyle{\int_{\la}^{\bar{a}}[\lambda-\sigma c(a)]\ell(a)dF_0(a)}$}&\scalebox{0.82}{$=(A(\la){-\la})\int_{\la}^{\bar{a}}\frac{\lambda-\sigma c(a)}{c(\la)-\lambda}dF_0(a)
+\int_{\la}^{\bar{a}}\int_a^{\bar{a}}\frac{\lambda-\sigma c(x)}{c(a)-\lambda}dF_0(x)dA(a)$}\nonumber\\
&\scalebox{0.82}{$={K+}A(\la)\int_{\la}^{\bar{a}}\frac{\lambda-\sigma c(a)}{c(\la)-\lambda}dF_0(a)
+\int_{\la}^{\bar{a}}\int_a^{\bar{a}}\frac{\lambda-\sigma c(x)}{c(a)-\lambda}dF_0(x)dA(a)$}\label{apartsc}
\end{align}
{where $K=-\la\int_{\la}^{\bar{a}}\frac{\lambda-\sigma c(a)}{c(\la)-\lambda}dF_0(a)$. }Combining (\ref{apartsa}) and (\ref{apartsc}), 
\begin{align}
\scalebox{0.82}{$\displaystyle{\int_{\la}^{\bar{a}}\left[A(a)+(\lambda-\sigma c(a))\ell(a)\right] dF_0(a)}$}&\scalebox{0.82}{$\displaystyle{={ K+}A(\la)\int_{\la}^{\bar{a}}\frac{c(\la)-\sigma c(a)}{c(\la)-\lambda}dF_0(a)
+\int_{\la}^{\bar{a}}\int_a^{\bar{a}}\frac{c(a)-\sigma c(x)}{c(a)-\lambda}dF_0(x)dA(a)}$}
\nonumber\\
&\scalebox{0.82}{$\displaystyle{={K+}A(\la)[1-S(\la)]+\int_{\la}^{\bar{a}}[1-S(a)]dA(a)}$},
\label{apartsd}
\end{align}
where $S$ is defined by (\ref{eq:G})
\begin{equation*}
S(a)=F_0(a) + \int_a^{\bar{a}}\frac{\sigma c(x)-\lambda}{c(a)-\lambda}dF_0(x).
\end{equation*}
Note that $S$ is continuous, $S(\la)$ is finite and $S(\bar{a})=1$. Also remark that $F_0(\la)=0$, so that $S(\la)$ is strictly negative. 

We claim that $S$ has bounded variation. Define $\Delta^+ (x) \equiv \max \{ \lambda-\sigma c(x), 0\}$ and $\Delta^{\dagger} (x) \equiv - \min \{ \lambda-\sigma c(x), 0\}$. Then, by (\ref{eq:G}):
\begin{align*}
\scalebox{0.8}{$\displaystyle{S(a)}$} & \scalebox{0.8}{$\displaystyle{= F_0(a) + \int_a^{\bar{a}}\frac{\Delta^+(x)}{c(a)-\lambda}dF_0(x) - \int_a^{\bar{a}}\frac{\Delta^{\dagger}(x)}{c(a)-\lambda}dF_0(x)}$}\\
& \scalebox{0.8}{$\displaystyle{= F_0(a) + \int_{\la}^{\bar{a}}\frac{\Delta^+(x)}{c(a)-\lambda}dF_0(x) - \int_{\la}^{a}\frac{\Delta^+(x)}{c(a)-\lambda}dF_0(x) - \int_{\la}^{\bar{a}}\frac{\Delta^{\dagger}(x)}{c(a)-\lambda}dF_0(x) + \int_{\la}^{a}\frac{\Delta^{\dagger}(x)}{c(a)-\lambda}dF_0(x)}$}.
\end{align*}
The first term on the right hand side of this equation is a cdf, nondecreasing in $a$. Consider each of the four integrals (without the sign that precedes them). Each integrand is a nonnegative-valued function (because $c(a) > \lambda $, and $\Delta^+$ and $\Delta^{\dagger}$ are each nonnegative), and each is nondecreasing in $a$ (because $c(a)$ declines in $a$). Therefore, each integral is nondecreasing in $a$. It follows that $S$ can be written as the sum/difference of five nondecreasing functions and consequently is of bounded variation. Therefore integration with respect to $S$ is well-defined. Define $P = A$ and $Q(a) = 1-S(a)$, and apply Lemma \ref{parts} yet again to  (\ref{apartsd}) to obtain (\ref{eq:obj}). 

The remainder now follows by simply applying Theorem \ref{th:1} to the induced problem $(S, R)$, solving for an optimal $A^*$, and then backing out the optimal learning problem via Observation \ref{obs:ic}.

\textbf{Case 2. }Now assume $\lambda \leqslant\s \int_{\la}^{\bar a}c(a)dF_0(a)$, so that $\ell^*(\la)=0$. 

Equations (\ref{apartsa}) and (\ref{apartsb}) still hold as in Case 1. But now note that $\ell(\la)=0$, and so, by setting $S(\la)=0$, we can rewrite (\ref{apartsb}) as
\begin{align}
\label{apartse}
\int_{\la}^{\bar{a}}[\lambda-\sigma c(a)]\ell(a)dF_0(a)=-A(\la)S(\la)
+\int_{\la}^{\bar{a}}\int_a^{\bar{a}}\frac{\lambda-\sigma c(x)}{c(a)-\lambda}dF_0(x)dA(a)
\end{align}

Finally, combine (\ref{apartsa}) and (\ref{apartse}) to again get
\begin{align}
\int_{\la}^{\bar{a}}\left[A(a)+(\lambda-\sigma c(a))\ell(a)\right] dF_0(a)=A(\la)[1-S(\la)]+\int_{\la}^{\bar{a}}[1-S(a)]dA(a)
\label{apartsf}
\end{align}

We prove that $S$ has bounded variation just as before. Since $S$ is left-continuous and only discontinuous at $\la$, integration with respect to $S$ is still well-defined. Define $P = A$ and $Q(a) = 1-S(a)$, apply Lemma \ref{parts} yet again  to  (\ref{apartsf}), and set $K=0$, to obtain (\ref{eq:obj}).

\qed

\emph{Proof of Proposition \ref{pr:fullp}.}
Proposition \ref{pr:pool} tells us that (\ref{eq:po}) is a sufficient condition for full pooling. Parts (i)--(iii) follow right away. \qed

\emph{Proof of Proposition \ref{pr:sup}.} Proposition \ref{app-pr:existence} in the Supplementary Appendix shows that an minimal optimal categorization exists, in the sense that no pooling interval can be split into any sub-categorization without reducing sender value. Fix such a categorization $A^*$. For any optimal categorization in the linear case, let $[a,a')$ be a separating interval. The proposition is equivalent to the claim that every pooling interval $[b,b')$ under $A^*$, $[a,a')\cap[b,b') =  \emptyset$. 

Suppose the claim is false. Then. by the optimality of $[a,a')$ in the linear case, it cannot be that  $R$ fosd $S$ conditional on $[b,b')$. For if it did, then $[b,b')$ would be contained in a pooling interval in any optimal categorization in the linear case, which contradicts the fact that $[a,a')$ is separating. Therefore, there is $c\in[b,b')$ such that
\begin{equation}
\kappa_S \equiv \frac{S(c)-S(b)}{S(b')-S(b)}\leqslant\frac{R(c)-R(b)}{R(b')-R(b)} \equiv \kappa_R
\label{eq:c}
\end{equation}
Using this information, we see that
\begin{align*}
\scalebox{0.8}{$\displaystyle{U\left(\int_b^{b'}x\frac{dR(x)}{R(b')-R(b)}\right)}$}
&\scalebox{0.8}{$\displaystyle{\leq \frac{R(c)-R(b)}{R(b')-R(b)}U\left(\int_b^cx\frac{dR(x)}{R(c)-R(b)}\right)+\frac{R(b')-R(c)}{R(b')-R(b)}U\left(\int_c^{b'}x\frac{dR(x)}{R(b')-R(c)}\right)}$}\\
&\scalebox{0.9}{$\displaystyle{=\kappa_RU\left(\int_b^cx\frac{dR(x)}{R(c)-R(b)}\right)+(1-\kappa_R)U\left(\int_c^{b'}x\frac{dR(x)}{R(b')-R(c)}\right)}$}\\
&\leq \scalebox{0.9}{$\displaystyle{\kappa_SU\left(\int_b^cx\frac{dR(x)}{R(c)-R(b)}\right)+(1-\kappa_S)U\left(\int_c^{b'}x\frac{dR(x)}{R(b')-R(c)}\right)}$}\\
&\scalebox{0.8}{$\displaystyle{=\frac{S(c)-S(b)}{S(b')-S(b)}U\left(\int_b^cx\frac{dR(x)}{R(c)-R(b)}\right)+\frac{S(b')-S(c)}{S(b')-S(b)}U\left(\int_c^{b'}x\frac{dR(x)}{R(b')-R(c)}\right)}$},
\end{align*}
where the first inequality comes from the assumed convexity of $U$, and the second from (\ref{eq:c}) and the fact that $U\left(\int_b^cx\frac{dR(x)}{R(c)-R(b)}\right)\leqslant U\left(\int_c^{b'}x\frac{dR(x)}{R(b')-R(c)}\right)$. But this inequality implies that sender value is not reduced by a split of the pooling interval $[b, b')$, which contradicts our choice of $A^*$. \qed

\section*{References}

Alonso, R. and O. Camara (2016a), `` "Bayesian Persuasion with Heterogeneous Priors," \emph{Journal of Economic Theory} \textbf{165}, 672--706.

Alonso, R., and O. Camara (2016b), ``Political Disagreement and Information in Elections,'' \emph{Games and Economic Behavior} \textbf{100}: 390--412.

Boleslavsky, R. and C. Cotton (2015), ``Grading Standards and Education Quality,"
\emph{American Economic Journal: Microeconomics} \textbf{7},  248--279. 

Boleslavsky, R. and K. Kim (2020), ``Bayesian Persuasion and Moral Hazard,''  \emph{working paper}.

Che, Y-K. and N. Kartik (2009), ``Opinions as Incentives,'' \emph{Journal of Political Economy}, \textbf{117}: 815--860.

Doval, L. and V. Skreta (2018), ``Constrained Information Design: Toolkit,'' mimeo. CalTech.

Dubey, P., and J. Geanakoplos (2010), ``Grading Exams: 100, 99, 98,? or A, B, C?," \emph{Games and Economic Behavior} \textbf{69}, 72--94.

Dworczak, P. and G. Martini (2019), ``The Simple Economics of Optimal Persuasion," \emph{Journal of Political Economy} \textbf{127}, 1993--2048.

Galperti, S. (2019), ``Persuasion: The Art of Changing Worldviews,'' \emph{American Economic Review} \textbf{109}, 996--1031.

Guo, Y., and E. Shmaya (2019), ``The Interval Structure of Optimal Disclosure," \emph{Econometrica} \textbf{87}, 653--675.

Kamenica, E. and M. Gentzkow (2011), ``Bayesian Persuasion," \emph{American
Economic Review} \textbf{101}, 2590--2615.

Kartik, N., F. Lee, and W. Suen (2016), ``Investment in Concealable Information by Biased Experts,'' \emph{The RAND Journal of Economics}, \textbf{48}: 24--43.

Kartik, N., F. Lee, and W. Suen (2020), ``Information Validates the Prior: A Theorem on Bayesian Updating and Applications,'' \emph{AER Insights}, forthcoming.

Kolotilin, A. (2018), ``Optimal Information Disclosure: A Linear Programming Approach," \emph{Theoretical Economics} \textbf{13}, 607--636.

Kolotilin, A. and H. Li (2019), ``Relational Communication,'' \emph{Theoretical Economics}, forthcoming.

Kolotilin, A., T. Mylovanov, and A. Zapechelnyuk (2019), ``Censorship as Optimal Persuasion,'' mimeo. 

Kolotilin, A. and A. Zapechelnyuk (2019), ``Persuasion Meets Delegation,'' \emph{working paper}.

Kolotilin, A. and A. Wolitzky (2020), ``Assortative Information Disclosure,'' \emph{working paper}.

Lipnowski, E., Mathevet, L. and D. Wei (2020), ``Attention Management," \emph{AER Insights}, \textbf{2}: 17--32.

Lipnowski, E. and D. Ravid (2020), ``Cheap Talk with Transparent Motives,'' \emph{Econometrica}, \textbf{88}: 1631--1660.

Lizzeri, A. (1999), ``Information Revelation and Certification Intermediaries," \emph{RAND Journal of Economics} \textbf{30}, 214--231.

Moldovanu, B., Sela, A. and X. Shi (2007), ``Contests for Status,'' \emph{Journal of Political Economy} \textbf{115}, 338--363.

Ostrovsky, M. and M. Schwarz (2010), ``Information Disclosure and Unraveling in
Matching Markets," \emph{American Economic Journal: Microeconomics} \textbf{2}, 34--63.

Popov, S. and D. Bernhardt (2013), ``University Competition, Grading Standards, and Grade Inflation," \emph{Economic Inquiry} \textbf{51}, 1764--1778.

Rayo, L. (2013), ``Monopolistic Signal Provision,'' \emph{The BE Journal of Theoretical Economics} \textbf{13}, 27--58.

Rayo, L., and I. Segal (2010), ``Optimal Information Disclosure." \emph{Journal of Political Economy} \textbf{118}, 949--987.

Rodina, D. and J. Farragout (2016), ``Inducing Effort through Grades,'' \emph{working paper}.

Saeedi, M. and A. Shourideh (2020), ``Optimal Rating Design,'' \emph{working paper}.

Shaked, M. and J. Shanthikumar (2007), ``Stochastic Orders,'' \emph{Springer Science \& Business Media}.

Skreta, V. and L. Veldkamp (2009), ``Ratings Shopping and Asset Complexity: A Theory of Ratings Inflation," \emph{Journal of Monetary Economics} \textbf{56}, 678--695.

Van den Steen, E. (2004), ``Rational Overoptimism (and Other Biases),'' \emph{American Economic Review}, \textbf{94}: 1141--1151.

Van den Steen, E. (2009), ``Authority Versus Persuasion,'' \emph{American Economic Review, Papers and Proceedings}, \textbf{99}: 448--53.

Van den Steen, E. (2010), ``On the Origin of Shared Beliefs (and Corporate Culture),'' \emph{ RAND Journal of Economics}, \textbf{41}: 617--648.

Van den Steen, E. (2011), ``Overconfidence by Bayesian-Rational Agents,'' \emph{Management Science}, \textbf{57}: 884--896.

Yildiz, M. (2004), ``Waiting to Persuade,'' \emph{Quarterly Journal of Economics}, \textbf{119}: 223--248.

\end{document}